\documentclass[letterpaper,11pt]{article}

\usepackage[utf8]{inputenc}
\usepackage[dvipsnames]{xcolor}
\usepackage{float,graphicx,verbatim,fullpage,hyperref,amssymb,amsmath,amsthm,enumerate,multicol, xspace,mathtools,thmtools,thm-restate,cleveref,xspace,tikz,caption,subcaption,wasysym, enumitem}
\usepackage[margin=1in]{geometry}
\usepackage{algorithm, cite}
\usepackage[noend]{algpseudocode}
\usepackage{enumitem}
\usepackage{comment}
\usepackage{setspace}
\usepackage{array}
\usepackage{stmaryrd}
\usepackage{xcolor}
\providecommand{\email}[1]{\href{mailto:#1}{\nolinkurl{#1}\xspace}}

\usetikzlibrary{calc}

\newcolumntype{C}{>{\centering\arraybackslash}p{0.55cm}}

\newcolumntype{D}{>{\centering\arraybackslash}p{0.6cm}}

\usepackage{mleftright}
\usepackage{hyperref}
    \hypersetup{ colorlinks=true, linkcolor=blue, filecolor=magenta, urlcolor=blue, citecolor=red}

\usetikzlibrary{arrows.meta}
\tikzset{>={Latex[width=1.5mm,length=1.5mm]}}

\tikzstyle{vertex}=[circle, draw, inner sep=1pt, minimum size=18pt]

\newfloat{procedure}{htbp}{loa}
\floatname{procedure}{Procedure}

\def\final{1}  
\def\iflong{\iffalse}
\ifnum\final=0  
\newcommand{\anote}[1]{{\color{blue}[{\small Alex: \bf #1}]\marginpar{\color{red}*}}}
\newcommand{\ynote}[1]{{\color{red}[{\small Young-San: \bf #1}]\marginpar{\color{red}*}}}
\newcommand{\todo}[1]{{\color{red}[{ TODO: \bf #1}]\marginpar{\color{red}*}}}
\else 
\newcommand{\anote}[1]{}
\newcommand{\ynote}[1]{}
\newcommand{\todo}[1]{}
\fi  

\def\R{\mathbb{R}}

\newcommand{\cP}{\mathcal{P}}

\newcommand{\cT}{\mathcal{T}}
\newcommand{\cH}{\mathcal{H}}

\def\ep{\varepsilon}

\def\colorforcode{magenta}

\newenvironment{claimproof}{{\emph{Proof of claim:}}}{\hfill$\square$}

\newtheorem{theorem}{Theorem}[section]

\newtheorem{lemma}[theorem]{Lemma}
\newtheorem{claim}[theorem]{Claim}

\newtheorem{observation}{Observation}

\theoremstyle{definition}
\newtheorem{definition}[theorem]{Definition}

\newcommand*{\myrulefill}[3][]{%
  \makebox[#2]{#1%
    \leaders\hrule height \dimexpr.5ex+.2pt\relax depth \dimexpr -.5ex+.2pt\relax \hfill
    \enskip{#3}\enskip
    \leaders\hrule height \dimexpr.5ex+.2pt\relax depth \dimexpr -.5ex+.2pt\relax \hfill\kern0pt}
}

\begin{document}

\title{Improved and Parameterized Algorithms for Online Multi-level Aggregation: A Memory-based Approach}
\author{Alexander Turoczy\thanks{University of Oxford. 
E-mail: \email{alexander.turoczy@cs.ox.ac.uk}.} \and Young-San Lin\thanks{Melbourne Business School. 
 E-mail: \email{y.lin@mbs.edu}.}}
\date{\today}

\maketitle

\begin{abstract}
    We study the online multi-level aggregation problem with deadlines (MLAP-D) introduced by Bienkowski, B\"{o}hm, Byrka, Chrobak, D\"{u}rr, Folwarczn\'{y},  Je\.{z}, Sgall, Thang, and Vesel\'{y} (ESA 2016, OR 2020). In this problem, requests arrive over time at the vertices of a given vertex-weighted tree, and each request has a deadline that it must be served by. The cost of serving a request equals the cost of a path from the root to the vertex where the request resides. Instead of serving each request individually, requests can be \emph{aggregated} and served by transmitting a subtree from the root that spans the vertices on which the requests reside, to potentially be more cost-effective. The aggregated cost is the weight of the transmission subtree. The goal of MLAP-D is to find an aggregation solution that minimizes the total cost while serving all requests. MLAP-D generalizes some well-studied problems including the TCP acknowledgment problem and the joint replenishment problem, and arises in natural scenarios such as multi-casting, sensor networks, and supply chain management.

    We present improved and parameterized algorithms for MLAP-D. Our result is twofold. First, we present an $e(D+1)$-competitive algorithm where $D$ is the depth of the tree.
    Second, we present an $e(4H+2)$-competitive algorithm where $H$ is the \emph{caterpillar dimension} of the tree. Here, $H \le D$ and $H \le \log_2 |V|$ where $|V|$ is the number of vertices in the given tree. The caterpillar dimension remains constant for rich but simple classes of trees, such as line graphs ($H=1$), caterpillar graphs ($H=2$), and lobster graphs ($H=3$). To the best of our knowledge, this is the first online algorithm parameterized on a measure better than depth. The state-of-the-art online algorithms are $6(D+1)$-competitive by Buchbinder, Feldman, Naor, and Talmon (SODA 2017) and $O(\log |V|)$-competitive by Azar and Touitou (FOCS 2020). Our framework outperforms the state-of-the-art ratios when $H = o(\min\{D,\log_2 |V|\})$.
    Our \emph{memory-based} algorithms extend transmission subtrees with a cost comparable to transmission subtrees used to serve previous requests. Our simple framework directly applies to trees with \emph{any} structure and differs from the previous frameworks that reduce the problem to trees with specific structures.
\end{abstract}

\thispagestyle{empty}
\newpage

\setcounter{page}{1}

\section{Introduction}

Aggregation optimization problems have a wide range of applications including supply chain management \cite{arkin1989computational,crowston1973dynamic,kimms2012multi,wagner1958dynamic}, communication networks \cite{dooly1998tcp,karlin2003dynamic}, and multicasting \cite{badrinath2000gathercast,bortnikov2001schemes}. The common scenario of these applications is that individually resolving each request can be costly over time. Requests can be aggregated and served as batches under some aggregation constraints to lower service costs.

This work investigates the \emph{online multi-level aggregation problem with deadlines} (MLAP-D) introduced in \cite{bienkowski2016online},\footnote{See also \cite{bienkowski2020online} for the journal version.} which captures two characteristics of aggregation optimization problems. First, the service network is hierarchical and has a rooted tree structure. Second, the requests arrive in an online and dynamic fashion, and the decision-making is irrevocable in serving the requests.

\subsection{Online Multi-level Aggregation with Deadlines}

In MLAP-D, we are given an undirected tree $\cT=(V,E)$ rooted at $r \in V$. Each vertex $v \in V$ is associated with a cost defined by $c: V \to \R_{\ge 0}$. The tree $\cT$ and the vertex costs are given offline, in advance. A \emph{request} $\rho$ is associated with a vertex $v_\rho \in V(\cT)$ and a time interval $[a_\rho, d_\rho]$. Here, $v_\rho$ is the location of the request, and $a_\rho, d_\rho \in \R_{\ge 0}$ are the arrival time and the deadline of the request $\rho$, respectively. The time starts at zero, and the continuous time model is used. At time $a_\rho$, request $\rho$ arrives and the information of $\rho$ ($v_\rho$ and $d_\rho$) is revealed. The requests appear in an online fashion, so at time $t$, any information about requests $\rho'$ with $a_{\rho'} > t$ is not revealed. 
The output of an algorithm for MLAP-D is a collection $T_1, ..., T_\ell$ of subtrees of $\cT$ rooted at $r$ at different times $t_1, t_2, ... t_\ell$ where $t_1 < t_2 < ... < t_\ell$. A request $\rho$ is served if during the time interval $[a_\rho, d_\rho]$, a rooted subtree $T_i$ contains $v_\rho$. It is possible that more than one request is served at the same vertex and the same time. Every request $\rho$ must be served before the deadline $d_\rho$. 
It is noteworthy that we need not serve requests one at a time, as requests with overlapping time intervals can be served simultaneously as a batch, which can be more cost-effective. Let $\mathcal{S} = \{(T_i, t_i)\}_{i=1}^\ell$ denote a solution for MLAP-D. The cost of $\mathcal{S}$ is the total cost of the subtrees at different times. More specifically, let
\[c(\mathcal{S}) = \sum_{i =1}^\ell c(T_i) \text{ be the cost of $\mathcal{S}$ where } c(T_i) = \sum_{v \in V(T_i)}{c(v)} \text{ is the cost of $T_i$.}\]

Without loss of generality we assume that the arrival times and deadlines of all requests are distinct, as any instance can be converted to an instance of this form via an infinitesimal perturbation.

We note following \cite{buchbinder2017depth} that the vertex-weighted formulation captures the edge-weighted problem as follows. First, we can assume that in the edge-weighted instance, $r$ only has one adjacent vertex. If $r$ has multiple adjacent vertices, we can separate the problem into multiple instances in which $r$ has one adjacent vertex. Combining the solution of these separated instances gives a solution of the original instance without any extra cost. Next, we can move the cost of any edge into the endpoint vertex further away from $r$ and finally remove $r$ and set the vertex adjacent to $r$ as the new root.

Given a problem instance $\mathcal{I}$, and a deterministic online algorithm ALG, we use $\textsc{alg}(\mathcal{I})$ to denote the cost of the solution to $\mathcal{I}$ outputted by ALG. Using the classic online algorithm framework, we compare this cost to that of the offline optimal solution OPT on this instance, which we denote $\textsc{opt}(\mathcal{I})$. An algorithm is $k$-competitive if $\textsc{alg}(\mathcal{I}) \leq k\cdot \textsc{opt}(\mathcal{I})$ for all problem instances $\mathcal{I}$. For a fixed problem instance, we will simply use $\textsc{alg}$ and $\textsc{opt}$ to denote the costs. 

MLAP-D was presented as a unification of several aggregation problems on trees \cite{bienkowski2016online}. We defer the discussions about these special case problems in Section \ref{sec:rw}. The first non-trivial online algorithm for MLAP-D was $D^2 2^D$-competitive \cite{bienkowski2016online}, which was later improved to the state-of-the-art $6(D+1)$-competitive algorithm \cite{buchbinder2017depth}. Here, $D$ is the depth of the tree. There is a paper claiming a $D$-competitive algorithm \cite{mcmahan2021dcompetitivealgorithmmultilevelaggregation} which we cannot confirm the correctness of.
When $D=2$, the competitive ratio is 2 \cite{bienkowski2015approximation,bienkowski2014better}, which remains the highest known lower bound on the competitive ratio of MLAP-D for trees of fixed depth. The problem is NP-hard \cite{arkin1989computational} and APX-hard \cite{bienkowski2015approximation,nonner2009approximating}. Closing the competitive ratio gap between $2$ and $O(D)$ has been an intriguing open problem.


\subsection{Our Contributions}

In this work, we present improved competitiveness for MLAP-D. Our result is twofold. First, we present an improved $O(D)$-competitive algorithm. Second, we present an $O(H)$-competitive algorithm where $H$ is the \emph{caterpillar dimension} of the tree. Our competitive algorithms are based on a generic \emph{memory-based} framework which judiciously constructs an online solution according to historical information.

\subsubsection{Our $O(D)$-competitive algorithm}

In Section \ref{sec:eD}, we show the following theorem by our memory-based framework.

\begin{restatable}{theorem}{thmeD} \label{thm:eD}
    There exists a $(1 + \frac{1}{D})^{D} (D+1) \le e (D+1)$-competitive algorithm for MLAP-D.
\end{restatable}

We emphasize that our algorithm can be directly applied to \emph{any} tree structure. This makes the analysis direct, and differs from some previous work that relies on reductions to trees with specific structures, such as $L$-decreasing trees \cite{bienkowski2021new, buchbinder2017depth}.

\subsubsection{$O(H)$-competitiveness}

To potentially improve the competitive ratio for MLAP-D, we consider a preferable measure over depth, namely, the \emph{caterpillar dimension} introduced in \cite{matouvsek1999embedding}. The notion of \emph{heavy path decomposition} was introduced in \cite{sleator1983data}.
A heavy path decomposition of $\cT$ is a vertex-disjoint partition $\cP$ such that (1) each element in $\cP$ is a set of vertices that contains exactly one leaf vertex and forms a path in $\cT$ and (2) any two vertices in an element in $\cP$ must have an ancestor-descendant relationship. 
Recall that $r$ is the root vertex of $\cT$.
Consider a leaf vertex $u$ and the $r$-$u$ path.
Let $d(\cP,u)$ denote the number of elements in $\cP$ intersecting the $r$-$u$ path.
The \emph{dimension} of the decomposition $\cP$ is $dim(\cP):=\max_{u \in L(\cT)}\{d(\cP,u)\}$ where $L(\cT)$ denotes the set of leaf vertices in $\cT$. The caterpillar dimension $H$ of $\cT$ is the minimum dimension among all heavy path decompositions. See Figure \ref{fig:hpd} for an example.

\begin{figure}[h]
    \centering
    \begin{tikzpicture}
        \node[vertex] (a) at (0,5) {$r$};
        \node[vertex] (1) at (-4,4) {$v_a$};
        \node[vertex] (2) at (0,4) {$v_b$};
        \node[vertex] (3) at (4,4) {$v_c$};
        \node[vertex] (11) at (-6,3) {$v_d$};
        \node[vertex] (12) at (-4,3) {$v_e$};
        \node[vertex] (13) at (-2,3) {$v_f$};
        \node[vertex] (21) at (0,3) {$v_g$};
        \node[vertex] (111) at (-6,2) {$v_h$};
        \node[vertex] (121) at (-4.5,2) {$v_i$};
        \node[vertex] (122) at (-3.5,2) {$v_j$};
        \path (a) edge (1)
        (a) edge [dashed] (2)
        (a) edge [dashed] (3)
        (1) edge [dashed] (11)
        (1) edge (12)
        (1) edge [dashed] (13)
        (2) edge (21)
        (11) edge (111)
        (12) edge (121)
        (12) edge [dashed] (122)
        ;
    \end{tikzpicture}
    \caption{$\cT$ consists of the solid and dashed edges. Components connected by the solid edges define the elements of $\cP$. The heavy path decomposition $\cP$ of $\cT$ has six elements: $\{v_d,v_h\}$, $\{r,v_a,v_e,v_i\}$, $\{v_j\}$, $\{v_f\}$, $\{v_b,v_g\}$, and $\{v_c\}$. When the leaf vertex $u=v_h, v_j, v_f, v_g,$ or $v_c$, the $r$-$u$ path intersects two elements in $\cP$ which maximizes $d(\cP,u)$, so the dimension of $\cP$ is two. A decomposition with a lower dimension does not exist, so the caterpillar dimension of $\cT$ is two.}
    \label{fig:hpd}
\end{figure}
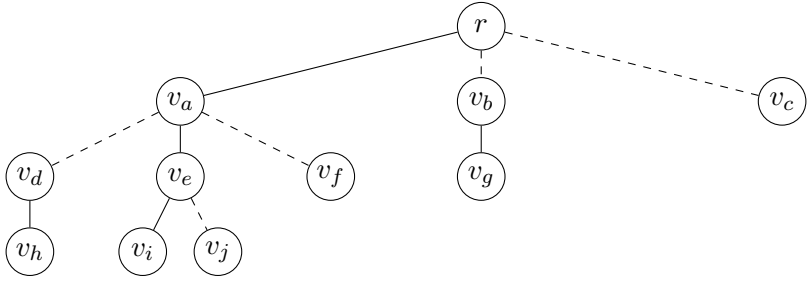

The caterpillar dimension remains constant for certain classes of trees, such as line graphs ($H=1$), caterpillar graphs ($H=2$), and lobster graphs ($H=3$). The caterpillar dimension is at most $D$ and $\log_2 |V|$, and can be computed in polynomial time, as well as the corresponding heavy path decomposition \cite{matouvsek1999embedding}.

In Section \ref{sec:4eH}, we show the following theorem using our memory-based framework.

\begin{restatable}{theorem}{thmeH} \label{thm:4eH}
    There exists a $(4H+2)(1+\frac{1}{2H+1})^{H+1}(1+\frac{1}{2H})^{H} \le e(4H+2)$-competitive algorithm for MLAP-D.
\end{restatable}

To the best of our knowledge, this is the first online algorithm parameterized on a measure better than depth. Interestingly, our framework outperforms the state-of-the-art ratios of $O(D)$ \cite{buchbinder2007online} and $O(\log |V|)$ \cite{azar2020beyond} if $H = o(\min\{D,\log_2 |V|\})$. For infinite line graphs ($H=1$), the best competitive ratio is known to be exactly 4 \cite{bienkowski2016online}.

\subsection{High-level Technical Overview}

\subsubsection{The $e(D+1)$-competitive algorithm}

In both the offline and online settings, the fundamental question when designing an algorithm for MLAP-D is to decide how many extra requests to aggregate into the service subtree, beyond the bare minimum required for a feasible solution. The biggest challenge one faces when developing a competitive algorithm is finding a way to lower-bound the performance of the optimum solution for a given instance. Most online algorithms for MLAP-D \cite{bienkowski2021new, buchbinder2017depth} have been \emph{memoryless}, in the sense that the algorithm's decision making is only a function of the current unserved requests and locations, and the relative order of their deadlines. This discards information which can potentially be used to justify further aggregation of the service subtree at a given time $t_i$. Our main algorithmic contribution is to incorporate memory naturally to inform the algorithm's decision-making process. From a mathematical standpoint, memory is used in the amortized analysis to ensure that we can attribute the cost of our algorithm $\textsc{alg}$ to the cost of the optimal solution $\textsc{opt}$ without double-counting, which is crucial for establishing a competitive ratio.

Leaving out some technical details, we describe the motivations behind our algorithms as follows. When an active request $\rho$ reaches its deadline at time $t_i$, we are forced to transmit a service subtree which contains $v_{\rho}$, and so our algorithm begins by initializing the service subtree as the path $r$-$v_{\rho}$. Each $v$ in the $r$-$v_{\rho}$ path adds a cost of $c(v)$ to the algorithm's solution, and therefore we want to delay making another service which includes $v$ until we are sure that $\textsc{opt}$ itself must have outputted services with costs comparable to $c(v)$. Therefore, it is natural that each node receives a budget proportional to $c(v)$ which is allocated towards adding extra services to the service subtree, and we call the nodes that the budget is spent on \emph{investments}. The investment procedure is handled by a recursive function, where deeper nodes are the first to make their investments. The difficulty with this approach is that at different times, several different vertices may invest in the same node, which could lead to over-counting the costs associated with $\textsc{opt}$.

We use memory to tackle this issue. Before vertices on the path from $r$-$v_\rho$ start spending their investment budget, they may also add their \emph{previous} investments to the service subtree and recursively expand the transmission subtree, in a stage of the algorithm we call the \emph{expansion stage}. This prevents nodes from reinvesting onto previous investments. The expansion stage is handled by a separate recursive process which occurs before the investment stage. It is then for these \emph{expansion} vertices that we make investments. We decide whether to add those previous investments as follows. After a vertex has made its investment, we set a timer in memory to record when all of the investments we have made are anticipated to reach their deadline. If a vertex $v$ at time $t_i$ is an expansion vertex again after the timer has elapsed, then we can make inferences about comparable costs $\textsc{opt}$ must have experienced based on our investments. In this case, we call $v$ an \emph{anticipated} vertex at time $t_i$. Otherwise, we can infer that there must have been a request that had appeared deeper than $v$ and expired before our timer elapsed, which we can also associate with a cost by $\textsc{opt}$. We call such vertices \emph{unanticipated} at time $t_i$. 

We briefly provide a broad overview of the analysis. In Section \ref{sec:ubalg}, we express $\textsc{alg}$ in terms of only the unanticipated expansion nodes. We use the standard amortized analysis techniques, in particular, an accounting argument to redistribute the costs in a top-down fashion based on proof by induction. Then in Section \ref{subsec:charge}, we show that the unanticipated expansion nodes serve as a way to lower bound $\textsc{opt}$. This requires using a charging scheme that defines a one-to-one mapping between the cost accounted for the unanticipated expansion vertices and $\textsc{opt}$. As a consequence of these two results, we can bound the competitive ratio of the algorithm. When we first define the algorithm, we leave a parameter which controls the extent of aggregation unspecified. Selecting the optimal parameter results in an $e(D+1)$-competitive algorithm for MLAP-D.

\subsubsection{The $e(4H+2)$-competitive algorithm}

Our $e(4H+2)$-competitive algorithm utilizes the framework developed by the $e(D+1)$-competitive algorithm. The main difference is the use of the heavy path decomposition structure while assigning memory-based investment budgets to vertices. Instead of assigning the budget for every vertex $v$ to be proportional to its cost $c(v)$, we allow one vertex per path in $\mathcal P$ to have a potentially much larger budget.

Following the same strategy as before, to prove the competitiveness of our algorithm, we show that (1) the cost $\textsc{alg}$ is comparable with the cost of the unanticipated expansion vertices and (2) we can relate the costs of unanticipated expansion vertices to $\textsc{opt}$ directly via a charging scheme. The charging scheme is identical to our $O(D)$-competitive algorithm. The main change to the analysis is to judiciously modify the redistribution accounting argument for the cost based on a proof by induction to show (1). Ultimately, by selecting the optimal parameters for the two different investment budgets, we obtain an $e(4H+2)$-competitive algorithm for MLAP-D.

\subsection{Related Work} \label{sec:rw}

\subsubsection{Online multi-level aggregation with waiting cost}

The most related variant of MLAP-D is the setting that each request has a \emph{waiting cost} instead of a strict deadline. The objective is to minimize the sum of the total service and waiting costs. This problem is termed the online multi-level aggregation problem (MLAP). For general MLAP, an algorithm with competitive ratio $D^4 2^D$ was presented in \cite{bienkowski2016online} and later improved to $O(D^2)$ \cite{azar2019general}. When $D=1$, MLAP is also called the \emph{TCP acknowledgment problem} (TCP-AP). The optimal competitive ratio for TCP-AP is 2 for deterministic algorithms \cite{dooly1998tcp} and $e/(e-1) \approx 1.582$ for randomized algorithms \cite{buchbinder2007online,karlin2003dynamic}. Interestingly, TCP-AP is equivalent to the \emph{lot sizing problem} \cite{wagner1958dynamic}, which has been extensively studied in the operations research community. When $D=2$, MLAP is also called the \emph{joint replenishment problem} (JRP). The state-of-the-art competitive ratio for JRP is 3 \cite{buchbinder2013online} and the current best lower bound is 2.754 \cite{bienkowski2014better} which improves an earlier bound of 2.64 \cite{buchbinder2013online}. When $H=1$, i.e., the input tree is a line, the upper bound for the competitive ratio is 5 \cite{bienkowski2013online} which improves an earlier bound of 8 \cite{brito2012competitive}, and the lower bound is 4 \cite{bienkowski2016online} which improves an earlier bound of $(5+\sqrt{5})/2 \approx 3.618$ \cite{bienkowski2013online}. \cite{azar2020beyond} introduced a unifying framework for problems with delay and deadlines admitting a competitive ratio of $O(\log |V|)$ for MLAP. A slight variant of MLAP was presented in \cite{khanna2002control} where the competitive ratio is logarithmic in the cost of the aggregation tree. We refer the reader to \cite{bienkowski2021new} for a concise summary for MLAP and MLAP-D.

\subsubsection{Offline multi-level aggregation} In offline MLAP-D, the information of all requests, including the location, the arrival time, and the deadline (also the waiting cost for offline MLAP), is given in advance. The offline problem currently has significantly better approximation results than the online problem. For offline MLAP-D, the current best approximation ratio is 2 \cite{becchetti2009latency,bienkowski2016online}. For offline MLAP, the current best approximation ratio is $2+\ep$ \cite{pc2013} by adapting an algorithm in \cite{levi2006primal} for the
multi-stage assembly problem. For offline TCP-AP, there is an optimal dynamic programming algorithm \cite{aggarwal1993improved}. Offline JRP is NP-hard \cite{arkin1989computational} and APX-hard \cite{bienkowski2015approximation,nonner2009approximating} and the current best approximation ratio is 1.791 \cite{bienkowski2014better} improving the previous ratios of 1.8 \cite{levi2008constant} and 2 \cite{levi2006primal}.

\subsubsection{Approximation algorithms parameterized on caterpillar dimension}

Quite surprisingly, the only algorithm parameterized on the caterpillar dimension that we are aware of is the advice complexity of the online $k$-server problem on trees \cite{renault2015online}. More specifically, to obtain optimal solutions for online $k$-server on trees, it is sufficient to use $2 + 2 \lceil \log_2 (H+1) \rceil$ bits of advice per request. We believe that the caterpillar dimension is a natural measure to evaluate the performance of approximation and online algorithms for challenging problems on trees, especially for problems like MLAP-D where the state-of-the-art competitive ratio is linear in the tree depth.

\section{The $e(D+1)$-competitive Algorithm} \label{sec:eD}

In this section, we prove Theorem \ref{thm:eD}.

We describe the components of Algorithm \ref{alg:eDalg}. Once a request $\rho$ reaches its deadline, the procedure $\textsc{OnDeadline}$ is called in order to construct a service subtree to serve $\rho$. During this procedure, vertices are added to a service subtree $T$ in two separate stages, which we call the \emph{expansion stage} and the \emph{investment stage}. Vertices added during the procedure \textsc{ExpansionStage} are added to the set $V^E$ and those added during \textsc{InvestmentStage} are added to $V^I$, and the final transmission tree is taken as the subtree induced on the disjoint union $V^E \cup V^I$. For each node $v$, the algorithm maintains three pieces of information which are used in both stages: $\ell_v \geq 0$, $next_v \geq 0$, and $I_v \subseteq V(\mathcal{T}_v) \setminus \{v\}$ where $\cT_v$ is the subtree of $\cT$ rooted at $v$ which includes all the vertices and edges under $v$. These variables are initialized in the \textsc{Initialization} procedure at the beginning of the problem instance as $\ell_v = c(v)$, $next_v = \infty$ and $I_v = \varnothing$. We later describe how this stored information is used.

\begin{algorithm}[h!]
\caption{$e(D+1)$-competitive algorithm for MLAP-D}\label{alg:eDalg}
\begin{algorithmic}[1]

\Procedure{Initialization}{} \label{proc:init}
\For{$v \in V(\mathcal T)$}
    \State $\ell_v \gets c(v)$
    \State $next_v \gets \infty$
    \State $I_v \gets \varnothing$ \Comment{Stores set of nodes which $v$ invested in most recently}
\EndFor
\State
\EndProcedure
\Procedure{OnDeadline}{} \Comment{Unserved request $\rho$ reaches deadline at time $t = d_{\rho}$}
    \State Initialize $V^E$ as the set of vertices on path from $r$ to $v_\rho$ \Comment{$V^E$ is a global variable}
    \State $V^E \gets $\textsc{ExpansionStage}($V^E$, $r$)
    \State Initialize $V^I \gets \varnothing$ \Comment{$V^I$ is a global variable}
    \State $V^I \gets $\textsc{InvestmentStage}($V^I$, $V^E$, $r$)
    \State \Return $T \gets \mathcal{T}[V^E \cup V^I]$ \Comment{Transmit $T$}
\EndProcedure
\State
\Procedure{ExpansionStage}{$V^E$, $v$} \
    \If{$t \geq next_v$} \Comment{$t$ denotes current time}
    \For{$v' \in I_v$}
        \State Add all vertices on the path from $v$ to $v'$ to $V^E$
    \EndFor
    \EndIf
    
    \For{child nodes $v'$ of $v$ in $V^E$}
        \State $V^E \gets $\textsc{ExpansionStage}($V^E$, $v'$)
    \EndFor
    \State \Return $V^E$
\EndProcedure
\State
\Procedure{InvestmentStage}{$V^I$, $V^E$, $v$} \

     \For{child nodes $v'$ of $v$ in $V^E$} \Comment{Invest for deeper nodes first}
            \State $V^I \gets $\textsc{InvestmentStage}($V^I$, $V^E$, $v'$)
    \EndFor

    \State $I_v \gets \varnothing$   \Comment{Begin investing into new set of nodes}
        
    
    \State $b_v \gets \theta \cdot c(v)$ \Comment{$\theta > 0$ is fixed parameter, for best competitive ratio set $\theta = D$}

    \While{$b_v > 0$ and there are pending requests on vertices in $V(\cT_v) \setminus (V^E \cup V^I)$} \label{line:investmentwhileloop}
        \State Let $\rho'$ be the pending request in $V(\cT_v) \setminus (V^E \cup V^I)$ with the earliest deadline
        \State Let $v'$ be first vertex on path from $v$ to $v_{\rho'}$ not in $V^E \cup V^I$ \Comment{$v' \in V(\cT_v)$} \label{line:ancestorsfirst}
        \State $V^I \gets \textsc{Invest}(v, v', V^I)$
    \EndWhile

    \State $next_v \gets \min \big(\{d_{\rho} : v_\rho \in V(\cT_v) \setminus (V^E \cup V^I) \} \cup  \{\infty \}\big)$ \label{alg:setnextinf}
    \State \Return $V^I$
\EndProcedure
\State
\Procedure{Invest}{$v$, $v'$, $V^I$} \Comment{Invest in $v'$ using $b_v$ and update $V^I$ if $v'$ is served}
\State $\Delta b_v \gets \min\{b_v, \ell_{v'} \}$ \Comment{Put as much of budget as possible towards buying $v'$}
\State $b_v \gets b_v - \Delta b_v$
\State $\ell_{v'} \gets \ell_{v'} - \Delta b_v$
\State $I_v \gets I_v \cup \{v'\}$ \Comment{Record that $v$ has invested in $v'$}
\If{$\ell_{v'} = 0$} \Comment{If $v'$ is served}
\State $V^I \gets V^I \cup \{ v'\}$
\State $\ell_{v'} \gets c(v')$
\EndIf
\State \Return $V^I$
\EndProcedure

\end{algorithmic}
\end{algorithm}

We first describe the expansion stage. When an unserved request $\rho$ on $v_{\rho}$ reaches its deadline at time $t$, we are forced to output a service subtree containing $v_{\rho}$. In this case, we say that the request $\rho$ is \emph{critical}. We start by adding the vertices on the path from $r$ to $v_{\rho}$ to $V^E$. From here, the algorithm recursively expands the transmission subtree by adding the vertices which are in $I_v$ to $V^E$ if $next_v \leq t$, starting with $v=r$ and then repeating this on its children in $V^E$. Observe for the first service occurring at $t_1$ that none of the sets $I_v$ are added to $V^E$, as $next_v = \infty$ initially.

We now describe the investment stage. Figuratively speaking, each vertex $v'$ has a price $c(v')$ which the algorithm must invest into $v'$ before $v'$ can be added to $V^I$. We keep track of the remaining cost for a vertex $v'$ using the variable $\ell_{v'}$; $\ell_{v'}$ starts at $c(v')$, is decreased with every investment, and is then reset to $c(v')$ after $v'$ has been added to $V^I$. For each node $v \in V^E$ (starting with the deepest nodes and working upwards), we allocate a budget proportional to $c(v)$ for the algorithm to put towards investing into vertices, and once $\ell_{v'}=0$, we add $v'$ to $V^I$. The bookkeeping for the spending of budget is done within the \textsc{Invest} procedure. The allocated budget for $v \in V^E$ is $\theta \cdot c(v)$, where for our competitive analysis, we find that the best competitive guarantee is obtained when $\theta = D$. An important property of the algorithm is that the budget of $v \in V^E$ is only spent on vertices deeper than $v$, and it invests towards the requests in non-decreasing order of deadline. The algorithm stores the set of vertices which $v \in V^E$ invested in into the variable $I_v$. Once a given node $v \in V^E$ has spent its budget, it sets $next_v$ as the earliest deadline among the remaining unserved requests beneath $v$, or to $\infty$ if there are no such active requests. Therefore, if the next time that $v$ is added during the expansion stage occurs at time $t'$, then if $t' \geq next_v$, the algorithm will also add $I_v$ during the expansion stage. 

To illustrate the execution of Algorithm \ref{alg:eDalg}, we provide an example in Appendix \ref{sec:app-ex}.

The following terms and notations will be used throughout the analysis.

\begin{definition}[Expansion vertices, Investment vertices]

Consider the transmission subtree $T_i = \mathcal{T}[V^E \cup V^I]$ transmitted by Algorithm \ref{alg:eDalg} at time $t_i$. A vertex $v$ is an \emph{expansion vertex at time} $t_i$ if $v \in V^E$, and an \emph{investment vertex at time} $t_i$ if $v \in V^I$. The set of expansion vertices at time $t_i$ will be denoted $V^E_i$, and the set of investment vertices at time $t_i$ will be denoted $V^I_i$.
    
\end{definition}

As the variables $I_v$ and $next_v$ change value within a single timestep, we make the following definition to refer to the values at specific times unambiguously.

\begin{definition}[$I_{v, i}, I_{v, i}', next_{v, i}, next_{v, i}'$] \label{def:next}
    
    If Algorithm \ref{alg:eDalg} transmits a service subtree at time $t_i$, we define $I_{v, i}$ (respectively, $I_{v, i}'$) to be the value of $I_v$ before (respectively, after) the function call \textsc{OnDeadline} at $t_i$. $next_{v, i}$ and $next_{v, i}'$ are defined analogously.
\end{definition}


Each set of expansion vertices $V^E_i$ is partitioned into two disjoint sets.

\begin{definition}[Anticipated vertices, unanticipated vertices]

We call a vertex $v \in V^E_i$ \emph{anticipated at time} $t_i$ if $next_{v, i} \leq t_i$. We call a vertex $v \in V^E_i$ \emph{unanticipated at time} $t_i$ if it is an expansion vertex at time $t_i$ and $next_{v, i} > t_i$. We denote the set of vertices anticipated at time $t_i$ as $A_i$, and those which are unanticipated as $\bar{A}_i$, so that $V_i^E = A_i \cup \bar{A}_i$ is a disjoint union.
\end{definition}

Our analysis will be split into two parts. First, we express $\textsc{alg}$ solely in terms of the nodes that are unanticipated. The following lemma is proved in Section \ref{sec:ubalg}.

\begin{restatable}{lemma}{lemam} \label{lem:amortized}
    \[\textsc{alg} \leq \left(1+\frac{1}{\theta}\right)^{D} (1+\theta)\sum_{i=1}^{\ell} c(\bar{A}_i).\]
\end{restatable}

In Section \ref{subsec:charge}, we use a charging argument to show that the costs of unanticipated expansion vertices across time can be charged directly to OPT without double-charging. We prove the following.
\begin{lemma} \label{lem:optlowerbound}
    \[\sum_{i=1}^{\ell}  c(\bar{A}_i) \leq \textsc{opt}.\]
\end{lemma}

Combining Lemmas \ref{lem:amortized} and \ref{lem:optlowerbound} implies
\[\textsc{alg} \leq \left(1+\frac{1}{\theta}\right)^{D} (1+\theta)\textsc{opt}.\]
The best choice is $\theta = D$ which implies Theorem \ref{thm:eD}. The proof is provided in Appendix \ref{app:pf-eD}.

Before proceeding with the proofs for Lemmas \ref{lem:amortized} and \ref{lem:optlowerbound}, we highlight some important properties of Algorithm \ref{alg:eDalg}, which we will refer to throughout the analysis. We give brief justifications of each after stating them.

\begin{observation} \label{obs:expectedsubtree}
    For all $t_i$, the subgraph of $\mathcal T$ induced by $V^E_i$ is a subtree rooted at $r$.
\end{observation}

Observation \ref{obs:expectedsubtree} follows because $r$ is added to $V^E$ during \textsc{OnDeadline}, and whenever we add $v$ to $V^E$, we add the entire $u$-$v$ path where $u \in V^E$ already, which maintains connectivity of $V^E$.

\begin{observation} \label{obs:deeperinvestments}
    For all $t_i$ and $v \in V(\mathcal{T})$, we have $I_{v, i} \subseteq V(\mathcal{T}_v)\setminus\{ v\}$.
\end{observation}

Observation \ref{obs:deeperinvestments} holds at $t_1$ because in \textsc{Initialization}, $I_v$ is set to $\varnothing$. The variable $I_v$ only changes at times $t_i$ when $v \in V^E_i$. Then, during an investment stage of $\textsc{RecursiveStep}(v)$ at $t_i$, we only invest in vertices in $\mathcal{T}_v$. Additionally, $v$ could not invest into itself as $v \in V^E_i$ already, and hence $I_v \subseteq V(\mathcal{T}_v)\setminus\{ v\}$ at all moments of the algorithm.

\begin{observation} \label{obs:ancestorsincluded}
    For all $t_i$, if $v \in V^E_i$ and $v' \in I_{v, i}'$, then all ancestors of $v'$ are in $T_i$.
\end{observation}

Observation \ref{obs:ancestorsincluded} is justified as follows. Consider $v \in V^E_i$ and $v' \in I_{v, i}'$. Then $\textsc{Invest}(v,v', V^I)$ must have been called for $b_v>0$ at time $t_i$. Due to line \ref{line:ancestorsfirst}, $v'$ was the first vertex on $v$-$v'$ path not in $V^E \cup V^I$, meaning that all ancestors of $v'$ must be in $T_i$ already.

\begin{observation} \label{obs:increasingnextupwards}
    For all $t_i$, if $u, v \in V_i^E$ and $u$ is an ancestor of $v$ with $I'_{u, i} \cap \mathcal{T}_v \neq \varnothing$, then $next_{u, i}' \geq next_{v, i}'$.
\end{observation}

Observation \ref{obs:increasingnextupwards} is justified as follows.  Assume for now that $next_{v, i}' \neq \infty$, so $next_{v, i}' = d_\rho$ where $\rho$ has the earliest deadline among unserved requests in $\mathcal{T}_v \setminus (V^E\cup V^I)$ after the investment stage at $v$ has completed at $t_i$. Observe that the investment stage for $v$ occurs before the investment stage for $u$, as $u$ is an ancestor of $v$. Therefore, when the investment stage for $u$ begins, the earliest deadline among unserved request in $\mathcal{T}_v \setminus (V^E\cup V^I)$ has a deadline at earliest $d_\rho$. Because the investment stage invests towards requests in non-decreasing order of deadline, it follows that if $I'_{u, i} \cap \mathcal{T}_v \neq \varnothing$, then at some point the request in $\mathcal{T}_u \setminus (V^E\cup V^I)$ with earliest deadline was located at a vertex in $\mathcal{T}_v$. But this implies that at the conclusion of the investment stage for $u$ that all requests in $\mathcal{T}_u \setminus (V^E\cup V^I)$ have deadline at earliest $d_{\rho}$. Therefore, $next_u$ would be set no earlier than $d_{\rho}$, and therefore $next_{u, i}' \geq d_{\rho} = next_{v, i}'$. For the other case, if $next_{v, i}' = \infty$, then after the investment stage for $v$ there were no requests in $\mathcal{T}_v \setminus (V^E\cup V^I)$, and so $I'_{u, i} \cap \mathcal{T}_v \neq \varnothing$ is impossible.

\subsection{Upper-bounding ALG} \label{sec:ubalg}

In this section, we prove Lemma \ref{lem:amortized}.


Observe that an investment vertex $v'$ at time $t_i$ is only added to $V_i^I$ if a total of $c(v')$ is invested by expansion vertices to reduce the variable $\ell_{v'}$ to zero. As each expansion vertex $v$ can invest no more than $\theta \cdot c(v)$, we obtain
\[
    \sum_{i=1}^{\ell} c(V_i^I) \leq \sum_{i=1}^{\ell} \theta \cdot c(V^E_i).
\]
This yields
\begin{equation}
    \textsc{alg} = \sum_{i=1}^{\ell} c(T_i) 
   = \sum_{i=1}^{\ell} [c(V_i^E) + c(V_i^I)]
    \leq \sum_{i=1}^{\ell} [c(V_i^E) + \theta \cdot c(V_i^E)] 
    = (1+\theta)\sum_{i=1}^{\ell}c(V_i^E).\label{eq:alltoexpansion}
\end{equation}

Now that we have expressed the cost solely in terms of the expansion vertices, we wish to express the cost solely in unanticipated expansion vertices, as these costs can then be directly associated with \textsc{opt} in Section \ref{subsec:charge}. To do so, we must move the costs of anticipated nodes to those which were unanticipated. We first show the following two lemmas. 


\begin{restatable}{lemma}{lemeae} \label{lemma:expectedafterexpansion}
Let $v \in A_i$ be an anticipated vertex at time $t_i$, then $v$ must have been an expansion vertex at a time $t_j < t_i$, in other words $v \in V^E_j$ for some $j < i$.
\end{restatable}

\begin{proof} 
    Suppose $v \in A_i$. By definition this means that $next_{v, i} \leq t_i$, and hence clearly $next_{v, i} \neq \infty$. As the $next_v$ variable is originally assigned to $\infty$ in the \textsc{Initialization} of the algorithm, it follows that $next_v$ must have had its value reassigned prior to $t_i$.  As the algorithm's variable $next_v$ is only reassigned at times for which $v$ is an expansion vertex, we therefore know that $v$ must have been an expansion vertex at some time before $t_i$. 
\end{proof}

\begin{restatable}{lemma}{lemubgeneric}  \label{lem:genericusedbudgets}
    Let $v \in A_i$ be an anticipated node at time $t_i$, and let $t_j$ be the latest time before $t_i$ that $v$ was an expansion vertex (see Lemma \ref{lemma:expectedafterexpansion}). Let $\theta c(v)$ be the budget assigned to $v$ in the investment stage for $v$. Then $c(I_{v, i}) \geq \theta c(v)$.
\end{restatable}

\begin{proof}
     Suppose $v \in A_i$. According to Lemma \ref{lemma:expectedafterexpansion}, we know that $v$ must have been an expansion vertex at a time earlier than $t_i$. Let $t_j$ be the most recent time before $t_i$ for which $v \in V^E_j$. 
     
     As the variable $next_v$ is not changed between $t_j$ and $t_i$, we know that $next_{v, j}' = next_{v, i}$, and hence $next_{v, j}' \neq \infty$ also. Therefore, we know that after leaving the while loop on line \ref{line:investmentwhileloop} within the $\textsc{InvestmentStage}(V^I, V^E, v)$ call at time $t_j$, it must be the case that there were still pending requests on vertices in $\mathcal{T}_v - (V^E \cup V^I)$, or else $next_{v, j}'$ would have been set to $\infty$ in line \ref{alg:setnextinf}. Thus, it must instead be the case that the condition $b_v > 0$ of the while loop failed, and so the entirety of the budget assigned to $v$ at time $t_j$ was spent. 
     Observe that each call of $\textsc{Invest}(v, v', V^I)$ will decrease the budget $b_v$ by $\Delta b_v$, but will add $v'$ to $I_{v}$, where $c(v') \geq \ell_{v'} \geq \Delta b_v$. Thus the amount by which the budget decreases is always at most the amount by which $c(I_v)$ increases. As during the iteration at time $t_j$, $b_v$ must have decreased by at least $\theta c(v)$, it follows that $c(I_v)$ must have increased by at least $\theta c(v)$ during the investment stage. As $I_v$ is set to $\varnothing$ at the beginning of the investment stage, we know $c(I_v) = 0$ initially, and thus if $c(I_v)$ increases by at least $\theta c(v)$ during the investment phase, we know $c(I_{v, j}') \geq \theta c(v)$. Observing that $c(I_{v, j}') = c(I_{v, i})$ finishes the proof.
     \end{proof}


We can now express the cost of all expansion nodes only in terms of the unanticipated nodes.
\begin{lemma} \label{lem:coretounexpected}

    For any $t_i$, we have

    \begin{equation*}
    c(V_i^E) = \sum_{v \in V^E_i} c(v) \leq \left(1+\frac{1}{\theta}\right)^{D} \sum_{\substack{v \in \bar{A}_i}}  c(v) = \left(1+\frac{1}{\theta}\right)^{D}c(\bar{A}_i)
\end{equation*}

\end{lemma}

\begin{proof}
    We proceed using the accounting method of amortized analysis. For a fixed $t_i$, consider placing coins of value $c(v)$ on each anticipated vertex $v \in V_i^E$. Thus the total cost initially on the structure $V^E_i$ equals the left-hand side of the inequality stated by the Lemma. We will redistribute the costs which are on anticipated nodes, to obtain a rearrangement of costs which is upper-bounded by the larger quantity of the lemma.
    
    We start with the root $r$. We \emph{process} $r$ as follows. If $r$ is an unanticipated node of $T_i$, we do not move the cost from $c(r)$ --- we will only move the costs off of anticipated nodes. Otherwise if $r$ is an anticipated node, then by Lemma \ref{lem:genericusedbudgets},  $c(r) \leq \frac{1}{\theta} c(I_{r, i}) = \frac{1}{\theta} \sum_{v' \in I_{r, i}} c(v')$. If $r$ is anticipated, then we know that all the nodes in $I_{r, i}$ were added to $V^E_i$, and so we can reallocate this cost of $c(r)$ from $r$ as an additional cost of at most $\frac{c(v')}{\theta}$ to each $v' \in I_{r,  i}$. 

    We descend the vertices in $V^E_i$ and repeat this process: we call a expansion node $v$ ``ready to redistribute" if all of its ancestors have been processed. To process an unanticipated $v \in \bar{A}_i$, we do nothing and the cost remains on $v$. For any anticipated $v \in A_i$  currently allocated a cost of $C \cdot c(v)$, we process $v$ by reallocating this as an additional cost of $\frac{C}{\theta} \cdot c(v')$ to each node $v'$ in $I_{v, i}$. Recall by Observation \ref{obs:deeperinvestments} that $v$ will only invest into vertices in $\mathcal{T}_v$, and so these costs are strictly pushed to deeper vertices.

    \begin{claim} \label{claim:inductiveredistribution}
        When a vertex $v \in V^E_i$ at depth $d$ is ready to redistribute but has not yet been processed, the allocated costs on vertices $v' \in V^E_i \cap \cT_v$ is at most $ (1+1/\theta)^dc(v')$.
    \end{claim}

    \begin{claimproof}
     We proceed by induction on $d$. The base case of $d=0$ is trivial, as this corresponds to the fact that all vertices in $v \in V^E_i$ have an allocated cost of $c(v)$ before the root $r$ is processed.

    Assume inductively that the claim has held true up to a depth $d$, and a vertex $v$ at depth $d+1$ is ready to redistribute. This means that when the parent of $v$, call it $u$, was ready to redistribute, all nodes $w \in V(\mathcal{T}_u)$ had allocated at most $(1+1/\theta)^d c(w)$. If $u$ is unanticipated, then the cost on $u$ is not moved, so none of the vertices in $\mathcal{T}_u$ receive any additional costs. This would imply that the costs on vertices $w$ in $\mathcal{T}_v$ would still be at most $(1+1/\theta)^d c(w) \leq (1+1/\theta)^{d+1} c(w)$, and so the claim will continue to hold after $u$ is processed. Otherwise, consider what would occur when $u$ was processed if $u$ was anticipated. An additional cost of at most $\frac{1}{\theta}(1+1/\theta)^d c(w)$ would be added to all vertices in $w \in \mathcal{T}_u$ (recall that the cost $C \cdot c(u)$ of $u$ is reallocated as an additional cost of $\frac{C}{\theta} \cdot c(w)$ to each node $w \in I_{u, i}$). Therefore immediately following the processing of $u$, each $w \in V(\mathcal{T}_u)$ would have had a cost of at most $(1+1/\theta)^d c(w) + \frac{1}{\theta}(1+1/\theta)^d c(w) = (1 + 1/\theta)^{d+1} c(w)$. At this point, $v$ would be ready to redistribute, and as $\cT_v$ is a subtree of $\cT_u$, the inductive step is finished and we have proven Claim \ref{claim:inductiveredistribution}.
    \end{claimproof}

    This redistribution process must terminate eventually, as the tree $\mathcal{T}$ has a bounded depth of $D$. Observe that this redistribution cost will leave anticipated vertices of $V^E_i$ with no costs, while an unanticipated vertex $v$ in $V^E_i$ at depth $d$ would be left with a cost of $(1+1/\theta)^d c(v) \leq (1+1/\theta)^{D}c(v)$. As we have at no point decreased the total cost on the structure $V^E_i$, the inequality stated by the lemma follows.
    \end{proof}

Combining inequality \eqref{eq:alltoexpansion} and Lemma \ref{lem:coretounexpected} implies Lemma \ref{lem:amortized}.

\subsection{Charging Scheme against OPT} \label{subsec:charge}

In this section, we prove Lemma \ref{lem:optlowerbound}.

We accomplish this using a charging argument to map the costs of unanticipated expansion vertices at times $t_i$ to OPT. However, before we can do so, we need to further characterize the unanticipated expansion nodes, which we accomplish in the following lemma.

\begin{lemma} \label{lem:unexpectedCharacterization}
    Let $\rho_i^*$ denote the pending request which triggered \textsc{OnDeadline} at $t_i$. Then for any $v \in \bar{A}_i$, either $v$ is on the $r$-$v_{\rho_i^*}$ path, or $v \in I_{v^*, i}$ for some $v^* \in A_i$. 
\end{lemma}

\begin{proof}
    Observe that there are only 2 cases in Algorithm \ref{alg:eDalg} for which a vertex $v$ can be added to $V^E_i$; (1) during \textsc{OnDeadline}, if $v$ is on the path $r$-$v_{\rho_i^*}$ and (2) during $\textsc{RecursiveStep}(u)$, if $v$ is on the path $u$-$u'$ where $u \in A_i$ and $u' \in I_{u, i}$. Case (1) clearly satisfies the lemma, so we will now assume the second case, that $v$ lies on the path $u$-$u'$ where $u \in A_i$ and $u' \in I_{u, i}$. As the lemma is obvious for $v=u'$, we continue under the assumption that $v \neq u'$.
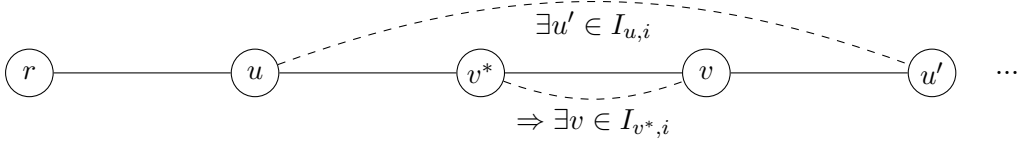
\begin{center}
    \begin{tikzpicture}
    \node[vertex] (r) at (0,0) {$r$};
    \node[vertex] (u) at (3,0) {$u$};
    \node[vertex] (vs) at (6,0) {$v^*$};
    \node[vertex] (v) at (9,0) {$v$};
    \node[vertex] (up) at (12,0) {$u'$};
    \node (deep) at (13,0) {$...$};
    \path (r) edge (u)
    (u) edge (vs)
    (vs) edge (v)
    (v) edge (up)
    (u) edge [dashed,bend left=20] node [below] {$\exists u' \in I_{u,i}$} (up)
    (vs) edge [dashed,bend right=20] node [below] {$\Rightarrow \exists v \in I_{v^*,i}$} (v);
\end{tikzpicture}
\end{center}

    By Lemma \ref{lemma:expectedafterexpansion}, there must be a time $t_{i^*} < t_i$ such that $u \in V^E_{i^*}$, let $t_{i^*}$ be the most recent such time, and observe that $next_{u, i^*}' = next_{u, i}$ and $I_{u, i^*}' = I_{u,i}$. As $u' \in I_{u, i} = I'_{u, i^*}$, by Observation \ref{obs:ancestorsincluded} we know that $v \in T_{i^*}$. We now make an important claim.

    \begin{claim} \label{claim:structural}
        For any $w$ on the path $u$-$u'$ (where $u \in A_i$, $u \in V_{i^*}^E$, and $u' \in I_{u, i}=I_{u, i^*}'$, and $t_{i^*}$ defined as before) if $w \in V_{i^*}^E$ then $w \in A_i$, $next_{w, i^*}'=next_{w, i}$ and $I_{w, i^*}'=I_{w, i}$.
    \end{claim}

    \begin{claimproof}
    By Observation \ref{obs:increasingnextupwards}, we have that $next_{w, i^*}' \leq next_{u, i^*}'$, as $w$ is on the $u$-$u'$ path. We know that $w \in V_{i^*}^E$, and we also know that $w \in V_i^E$, because $w$ is on the $u$-$u'$ path where $u, u' \in V^E_i$, and $V^E_i$ forms a subtree due to Observation \ref{obs:expectedsubtree}. Further observe that it is impossible for $w$ to be anticipated at any time $t$ with $t_{i^*} < t < t_i$, because this would necessitate that $u$ would be anticipated at time $t$ by Observation \ref{obs:expectedsubtree}, which would contradict the fact that $t_{i^*}$ is the latest time before $t_i$ for which $u$ is anticipated. With this observation, it follows that $next_{w, i^*}'=next_{w, i}$. Putting this together, we have that $next_{w, i} = next_{w, i^*}' \leq next_{u, i^*}' = next_{u, i} \leq t_i$, where we used the fact that $u \in A_i$. Hence, $w \in A_i$ by definition, proving the claim.     \end{claimproof}

    Earlier we found that $v \in T_{i^*} = V_{i^*}^E \cup V_{i^*}^I$. Due to Claim \ref{claim:structural}, we now know that $v \not \in V_{i^*}^E$ as otherwise $v \in A_i$, which would contradict the fact that $v$ is an unanticipated vertex at time $t_i$. 
    
    Thus $v \in V_{i^*}^I$, meaning there must have been some $v^* \in V^E_{i^*}$ such that $\textsc{Invest}(v^*, v, V^I)$ with $b_{v^*}>0$ was called at time $t_{i^*}$ and resulted in the addition of $v$ into $V_{i^*}^I$, and so $v \in I_{v^*, i^*}'$. Observe that by the time $\textsc{Invest}(u,u', V^I)$ was called, we would have needed all vertices on the path $u$-$u'$ already present in $V^E \cup V^I$, and hence $v^*$ must have invested in $v$ before $\textsc{Invest}(u,u', V^I)$ was called. Thus, $v^*$ must be at least as deep as $u$. $v^*$ must also be an ancestor of $v$ in order for $v^*$ to have invested into $v$ by Observation \ref{obs:deeperinvestments}, and hence $v^*$ is on the path $u$-$v$, and hence also on the path $u$-$u'$. Therefore Claim \ref{claim:structural} applies, meaning $v^* \in A_i$. Finally, $v \in I_{v^*, i^*}' = I_{v^*, i}$, finishing the proof.
    \end{proof}

We are now ready to define the charging scheme.

\subsubsection{The Charging Scheme}
Let $v \in \bar{A}_i$ be an unanticipated expansion vertex at time $t_i$. Let $\rho_i^*$ be the request which triggered the creation of a service at time $t_i$, so $d_{\rho_i^*} = t_i$. Using Lemma \ref{lem:unexpectedCharacterization}, we know that $v$ is either on the $r$-$v_{\rho_i^*}$ path, or $v \in I_{v^*, i}$ for $v^* \in A_i$. We thus split the charging scheme into two cases: 

\begin{itemize} 
    \item \emph{Case 1: $v$ is on the path $r$-$v_{\rho_i^*}$}. Then, OPT must have served $\rho_i^*$ within the time $[a_{\rho_i^*}, d_{\rho_i^*}] = [a_{\rho_i^*}, t_i]$. OPT's service to  $\rho_i^*$ must have also included $v$, and so we charge the cost of $c(v)$ against any of these $c(v)$ costs OPT incurs for including $v$ in its service history within $[a_{\rho_i^*}, t_i]$.
    
    \item \emph{Case 2: $v$ is not on the path $r$-$v_{\rho_i}$, and $v \in I_{v^*, i}$ for some $v^* \in A_i$}. If there are multiple such $v^*$, we simply fix one arbitrarily for each unanticipated $v \in \bar{A}_i$. By Lemma \ref{lemma:expectedafterexpansion} we know that $v^*$ must have been an expansion vertex at a time before $t_{i}$, so suppose that this occurred most recently at $t_{i^*} < t_i$. Then we would have that $next_{v^*, i^*}' = next_{v^*, i}$. As $v^* \in A_i$, we have $next_{v^*, i} \leq t_i$. Now, as $v \in I_{v^*, i} = I'_{v^*, i^*}$ and ALG processes requests in $\mathcal{T}_{v^*}$ in non-decreasing deadline, it follows that while ALG was investing the budget on $v^*$ at time $t_{i^*}$, there must have been an unserved request $\rho$ on a vertex $v_{\rho}$ such that $v$ was on the path $v^*$-$v_\rho$. Because $next_{v^*, i^*}'$ is determined as the next pending task in $\mathcal{T}_u$ not yet in $V^E_i \cup V_i^I$, we know that $d_\rho \leq next_{v^*, i^*}'$. Linking this together, we find that $d_\rho \leq next_{v^*, i^*}' = next_{v^*, i} \leq t_i$.
    
    Thus OPT must have served $\rho$ within the time $[a_{\rho}, d_{\rho}]$, which is a subset of $[a_\rho, t_i]$ by the above inequality. Such a service must have included $v$, and thus we charge $c(v)$ to the associated cost of $c(v)$ incurred by OPT during the timeframe $[a_{\rho}, t_i]$.
\end{itemize}

In both Case 1 and Case 2, we find ourselves charging $c(v)$ to a portion of OPT's service history within $[a_\rho, t_i]$ which serves $\rho$ where $v_{\rho}$ is deeper than $v$. In order to unify these two cases, we call this particular $\rho$ the request \underline{\emph{associated}} with unanticipated node $v$ at time $t_i$. 

Using the above charging scheme, the cost of any unanticipated vertex at a time $t_i$ is mapped to a cost in $\textsc{opt}$. To establish Lemma \ref{lem:optlowerbound}, we are left with showing that this mapping is one-to-one, i.e., no one portion of $\textsc{opt}$ is charged to more than once.

\begin{lemma} \label{lem:uniquecharges}
    The charging scheme for ALG consists of unique charges against OPT.
\end{lemma}

\begin{proof}
    Seeking a contradiction, suppose instead that two distinct charges for a vertex $v \in \bar{A}_i$ and $v \in \bar{A}_j$ coincide, say at times $t_j$ and $t_i$ where $j < i$. We will do a case analysis based on whether the charge at $t_i$ is Case 1 or Case 2. Suppose that $\rho_j$ and $\rho_i$ are the requests \emph{associated} (as defined above) to $v$ at times $t_j$ and $t_i$ respectively. The charges to OPT's history must have taken place between times $[a_{\rho_j}, t_j]$ and $[a_{\rho_i}, t_i]$. Thus, they could only coincide if $a_{\rho_i} < t_j$ (utilizing our assumption of distinct arrivals and deadlines for distinct services). In words, this means $\rho_i$ was present when $v$ was an unanticipated expansion vertex at time $t_j$. This leads to the following observation used throughout the proof.

    \begin{observation} \label{ob:nextvsd}
        For any $t_k > t_j$, if $\rho_i$ was not served before $t_k$, we must have $next_{v, k} \leq d_{\rho_i}$.
        
    \end{observation}
    
    We split the remainder of the proof into the two cases depending on the type of charging for $v$ which occurs at time $t_i$. We will show that both lead to contradiction.
    
    \underline{\textbf{Case 1}}. In this scenario, $v$ is on the critical path due to the expiry of $\rho_i$, meaning $\rho_i$ was not served before time $t_i$, and so $k=i$ satisfies the conditions required for Observation \ref{ob:nextvsd}. Thus $next_{v, i} \leq d_{\rho_i}$. As $d_{\rho_i} = t_i$, this would imply $next_{v, i} \leq t_i$ which would contradict the fact that $v \in \bar{A}_j$.

    \underline{\textbf{Case 2}}. This means at time $t_i$, for some ancestor $v^* \in A_i$ of $v$, that $v \in I_{v^*, i}$. Recall by definition this means that the last time $v^*$ was an expansion node prior to $t_i$, at time $t_{i^*} < t_i$, $\rho_i$ must have been active where $v_{\rho_i} \in \mathcal{T}_{v}$. We know that $v^* \in V^E_i$, as $v \in V^E_i$ and $V^E_i$ is a subtree by Observation \ref{obs:expectedsubtree}. Therefore it must be true that $t_{i^*} \geq t_j$, as $t_{i^*}$ is defined as the latest time before $t_j$ that $v^*$ was an expansion vertex. More strictly, $t_{i^*} > t_j$ as it is impossible for $v^*$ to invest into $v$ at time $t_j$ because $v \in V^E_j$, i.e it is an expansion vertex at time $t_j$. Hence we know $t_j < t_{i^*} < t_i$.

    We know that $\rho_i$ must not have been served before time $t_{i^*}$, because it was this service which caused $v^*$ to invest into $v$ at time $t_{i^*}$, by definition of Case 2 charge. Hence by Observation \ref{ob:nextvsd}, we know that $next_{v, i^*} \leq d_{\rho'}$ as $t_{i^*} > t_j$. Observe further that $v$ could not have been an expansion vertex at time $t_{i^*}$, otherwise $v^*$ would not have invested into $v$ at $t_{i^*}$. Therefore the variable $next_v$ would not be affected at time $t_{i^*}$, meaning $next_{v, i^*} = next_{v, i^*}'$ and thus $next_{v, i^*} \leq d_{\rho_i}$ still. It is impossible for $v$ to be an expansion vertex within the time interval $(t_{i^*}, t_i)$, because this would imply that $v^*$ would be an expansion vertex within $(t_{i^*}, t_i)$ by Observation \ref{obs:expectedsubtree}. Thus it follows that the variable $next_{v}$ is not changed in the interval $(t_{i^*}, t_i)$, and so $next_{v, i} \leq d_{\rho_i}$.
    
    However, from the discussion in the Case 2 charge definition, we already know that $d_{\rho_i} \leq t_i$, and so $next_{v, i} \leq d_{\rho_i} \leq t_i$. This contradicts the fact that $v \in 
    \bar{A}_i$, thereby finishing the proof.
\end{proof}

\section{The $(4H+2)e$-competitive Algorithm} \label{sec:4eH}

In this section, we prove Theorem \ref{thm:4eH}.


We show that Algorithm \ref{alg:eDalg} can be modified in a way which results in the competitive ratio being bounded above by a linear function of the caterpillar dimension $H$ of $\mathcal{T}$. Throughout this section, we assume that we are given a heavy path decomposition $\cP$ for $\cT$ with dimension $H$.

This section builds upon the work from Section \ref{sec:eD}, and thus we assume familiarity with Algorithm \ref{alg:eDalg} and the notations and definitions introduced previously. We also introduce some new pieces of notation.

\begin{definition}[$P(v)$ and $p(v)$]
    Let $v$ be a vertex in $\cT$. Let $P(v)$ denote the path which $v$ lies on in $\cP$. Let $p(v)$ be the intersection of the path $P(v)$ and the path from $r$ to $v$.
\end{definition}

\begin{algorithm}[h!]
\caption{$(4H+2)e$-competitive algorithm for MLAP-D}\label{alg:heavyalg}
\begin{algorithmic}[1]

\State --------------- The following procedures are identical to Algorithm \ref{alg:eDalg} ---------------
\Procedure{Initialization}{}
\EndProcedure
\Procedure{OnDeadline}{}
\EndProcedure
\Procedure{Invest}{$v$, $v'$, $V^I$}
\EndProcedure
\Procedure{ExpansionStage}{$V^E$, $v$}
\EndProcedure
\State
\Procedure{ModifiedInvestmentStage}{$V^I$, $V^E$, $v$} \Comment{Modified sections are \textcolor{\colorforcode}{colored}}

    \For{child nodes $v'$ of $v$ in $V^E$} \Comment{Invest for deeper nodes first}
            \State $V^I \gets $\textsc{ModifiedInvestmentStage}($V^I$, $V^E$, $v'$)
    \EndFor
    \State $I_v \gets \varnothing$  \label{line:init-0} 
    \Comment{Begin investing into new set of nodes}
        
    \State \textcolor{\colorforcode}{Let $v^\downarrow_{P(v)}$ denote the deepest vertex in $V^E \cap P(v)$}

    \If {\textcolor{\colorforcode}{$v = v^\downarrow_{P(v)}$}} \Comment{\textcolor{\colorforcode}{$\theta_1 > 0$ is parameter, for best competitive ratio set $\theta_1 = 2H+1$}}
        \State \textcolor{\colorforcode}{$b_v \gets \theta_1 \cdot c(p(v))$} \label{alg2:line-17}
    \Else \Comment{\textcolor{\colorforcode}{$\theta_2 > 0$ is parameter, for best competitive ratio set $\theta_2 = 2H$}}
        \State \textcolor{\colorforcode}{$b_v \gets \theta_2 \cdot c(v)$} 
    \EndIf

    \While{$b_v > 0$ and there are pending requests on vertices in $V(\cT_v) \setminus (V^E \cup V^I)$} \label{alg:loop}
        \State Let $\rho'$ be the pending request in $V(\cT_v) \setminus (V^E \cup V^I)$ with the earliest deadline.
        \State Let $v'$ be first vertex on path from $v$ to $v_{\rho'}$ not in $V^E \cup V^I$ \Comment{$v' \in V(\cT_v)$} 
        \If{\textcolor{\colorforcode}{$v \neq v^\downarrow_{P(v)}$ but $v' \in \mathcal{T}_{v^\downarrow_{P(v)}}$}} \label{line:cancelinvesting}
            \State \textcolor{\colorforcode}{$I_v \gets I_{v^\downarrow_{P(v)}}$} \Comment{\textcolor{\colorforcode}{Utilize investments from $v^\downarrow_{P(v)}$ instead}}
            \State \textcolor{\colorforcode}{Escape while-loop}
        \EndIf
        \State $V^I \gets \textsc{Invest}(v, v', V^I)$
    \EndWhile
    
    \State $next_v \gets \min \big(\{d_{\rho} : v_\rho \in V(\cT_v) \setminus (V^E \cup V^I) \} \cup  \{\infty \}\big)$
    \State \Return $V^I$
\EndProcedure

\end{algorithmic}
\end{algorithm}

The motivation behind Algorithm \ref{alg:heavyalg} is that if we wish for the competitive ratio to not be parameterized by the depth of $\mathcal{T}$, we cannot follow the accounting method of Lemma \ref{lem:coretounexpected}, because this explicitly descends the tree node-by-node and invariably leads to a function of $D$. If we wish to parameterize instead in terms of $H$, then one would need the accounting method to be performed path-by-path instead.

Algorithm \ref{alg:heavyalg} accomplishes this aim by allowing the deepest expansion node at time $t_i$ on each path $P \in \mathcal P$, which we will denote $v^\downarrow_{P, i}$, to potentially have a much larger budget to spend during the investment stage in comparison to Algorithm \ref{alg:eDalg}. Instead of the budget for $v^\downarrow_{P, i}$ being $\theta_1 \cdot c(v^\downarrow_{P, i})$ for some $\theta_1$ as was the case in Algorithm \ref{alg:eDalg}, for these nodes the budget at time $t_i$ becomes $\theta_1 \cdot c(p(v^\downarrow_{P, i}))$. Our algorithm will specified in terms of two parameters, $\theta_1$ and $\theta_2$, and we will choose their values to optimize our bound on the competitive ratio during the analysis. Additionally, $v^\downarrow_{P, i}$ is made to be the only vertex on $P(v^\downarrow_{P, i})$ which can make investments on vertices in $\mathcal{T}_{v^\downarrow_{P, i}}$. If $v \in P$ is not the deepest expansion node on $P(v)$, and if the next unserved request is within $\mathcal{T}_{v^\downarrow_{P(v), i}}$, then we stop making investments for $v$, which is done in line \ref{line:cancelinvesting}. The idea is that because $v^\downarrow_{P, i}$ has a disproportionally large budget, we can simply use the investments from $v^\downarrow_{P, i}$ for $v$ as well. 

\begin{definition}[$v^\downarrow_{P, i}$, Disproportional vertices, recently disproportional vertices]
We define $v^\downarrow_{P, i}$ to be the deepest node in $P \cap V^E_i$, where $P \in \mathcal P$. We say a vertex $v$ is \emph{disproportional} at time $t_i$ if $v = v_{P(v), i}^\downarrow$. That is, $v$ is the deepest expansion vertex at time $t_i$ in $P(v)$. Let $v \in A_i$, and let $t_j < t_i$ be the latest time before $t_i$ that $v$ was an expansion vertex (which must exist due to Lemma \ref{lemma:expectedafterexpansion}). Then we say $v$ was \emph{recently disproportional} at time $t_i$ if $v$ was disproportional at time $t_j$, i.e $v = v_{P(v), j}^\downarrow$.
\end{definition}




Similar to the proof of Theorem \ref{thm:eD}, our analysis is split into two parts. First, we express $\textsc{alg}$ in terms of the unanticipated nodes. The following lemma is proved in Section \ref{app:pf-am-h}.

\begin{restatable}{lemma}{lemamh} \label{lem:amortized-h}
    \[\textsc{alg} \leq \left(1+\frac{1}{\theta_1}\right)^{H+1} \left(1+\frac{1}{\theta_2}\right)^{H} (1+\theta_1+\theta_2)\sum_{i=1}^{\ell} c(\bar{A}_i).\]
\end{restatable}

Second, we use the same charging argument for Lemma \ref{lem:optlowerbound} to show that the costs of unanticipated expansion vertices across time can be charged directly to OPT without double-charging. Combining Lemmas \ref{lem:amortized-h} and \ref{lem:optlowerbound} implies
\[\textsc{alg} \leq \left(1+\frac{1}{\theta_1}\right)^{H+1} \left(1+\frac{1}{\theta_2}\right)^{H} (1+\theta_1+\theta_2)\textsc{opt}.\]
The best choice is $\theta_1 = 2H+1$ and $\theta_2 = 2H$ which implies Theorem \ref{thm:4eH}. The proof is provided in Appendix \ref{app:pf-4eH}.

Crucially, Algorithm \ref{alg:heavyalg} has built upon Algorithm \ref{alg:eDalg} using an approach which allows the 
key lower bound on \textsc{opt} provided by Lemma \ref{lem:optlowerbound} to continue to hold. To see why the analysis of Section \ref{subsec:charge} can be applied identically to Algorithm \ref{alg:heavyalg}, observe that Observations \ref{obs:expectedsubtree}, \ref{obs:deeperinvestments}, \ref{obs:ancestorsincluded} and \ref{obs:increasingnextupwards} are also true for Algorithm \ref{alg:heavyalg}. Then, because Section \ref{subsec:charge} only utilizes these observations rather than any features particular to either algorithm, we see that Lemma \ref{lem:optlowerbound} provides a robust lower bound for \textsc{opt}. Therefore, Algorithm \ref{alg:eDalg} can be used as a generic framework for MLAP-D algorithms.

\subsection{Upper-bounding ALG} \label{app:pf-am-h}

In this section, we prove Lemma \ref{lem:amortized-h}.

We begin by observing that the inclusion of each expansion node $v \in V^E_i$ adds at most $(\theta_1 + \theta_2)c(v)$ to additional budgets towards investments in total; a contribution of $\theta_2c(v)$ to the investment stage of $v$ and a contribution of $\theta_1 c(v)$ to the budget of $v^\downarrow_{P(v)}$. Thus we have 
\begin{equation*}
    (\theta_1 + \theta_2) \cdot \sum_{i=1}^{\ell} c(V^E_i) \geq  \sum_{i=1}^{\ell} c(V_i^I)
\end{equation*}
which yields 
\begin{equation}  \label{lem:alltoexpansion-h}
    \textsc{alg} = \sum_{i=1}^{\ell} c(T_i) = \sum_{i=1}^{\ell} [c(V_i^E) + c(V_i^I)] \leq \sum_{i=1}^{\ell} [c(V_i^E) + (\theta_1+\theta_2)c(V_i^E)] = (1+\theta_1+\theta_2)\sum_{i=1}^{\ell} c(V_i^E).
\end{equation}

Before proceeding, we define the concept of a \emph{heavy path trees} to simplify our analysis.

\begin{definition}[Heavy path tree]
    Given a heavy path decomposition $\cP$ for $\cT=(V,E)$, the heavy path tree $\cH$ is such that (1) $V(\cH) = \cP$, (2) the root of $\cH$ is the element in $\cP$ that contains the root vertex $r$, and (3) $E(\cH) = \{(P_1,P_2) \mid P_1, P_2 \in \cP \text{ and } \exists u \in P_1, v \in P_2 \text{ such that } (u,v) \in E\}$.
\end{definition}

Following the example in Figure \ref{fig:hpd}, the root is the $r$-$v_i$ path, and the children of the root are the $v_d$-$v_h$ and $v_b$-$v_g$ paths and the singleton paths $v_j$, $v_f$, and $v_c$. We note that the depth of the heavy path tree is the same as the dimension of the heavy path decomposition.

The next two lemmas provide lower bounds on investments, depending on whether there exists a recently disproportional vertex on a given path $P(v)$.

\begin{lemma} \label{lem:disproportionalvertices}
    For any $t_i$, let $v^* \in A_i$ be a recently disproportional vertex. Then $c(I_{v^*, i}) \geq \theta_1 \cdot c(p(v^*))$. \end{lemma}
\begin{proof}
    Suppose we have such a $v^*$. Let $t_j$ be the most recent time before $t_i$ that $v^*$ was an expansion vertex. Then because $v^*$ was disproportional at time $t_j$, we know that at $t_j$, $v^*$ received a budget of $\theta_1 \cdot c(p(v^*))$. Observe that as $v^*$ is disproportional at time $t_j$, that at $t_j$ the while loop of Algorithm \ref{alg:heavyalg} is effectively identical to the while loop of Algorithm \ref{alg:eDalg}, because the if-block at line \ref{line:cancelinvesting} will never be entered. Therefore, we can apply Lemma \ref{lem:genericusedbudgets}, and hence $c(I_{v^*, i}) \geq \theta_1 \cdot c(p(v^*))$.
\end{proof}

\begin{lemma} \label{lem:nodisproportionalbelow}
    For any $t_i$, let $v \in A_i$ be a vertex such that there are no recently disproportional vertices in $P(v) \cap \mathcal T_v$ at time $t_i$. Then $c(I_{v, i}) \geq \theta_2 \cdot c(v)$,
\end{lemma}

\begin{proof}
    As $v \in A_i$, by Lemma \ref{lemma:expectedafterexpansion}, there must have been a time $t_j$ prior to $t_i$ for which $v$ was an expansion vertex. Let $t_j$ be the latest such time, so $I_{v, i} = I_{v, j}'$. We claim that during the investment stage for $v$ at $t_j$, the if-block at line \ref{line:cancelinvesting} is never entered. This follows quickly, for if it were entered then $I_{v, j}' = I_{v^\downarrow_{P(v), j}, j}'$, which would imply $v^\downarrow_{P(v), j}$ would be added to $V_i^E$ during the expansion stage. However, $v^\downarrow_{P(v), j}$ would be a recently disproportional vertex at time $t_i$ within $\mathcal T_v \cap P$ which would contradict the definition of $v$. Now, if the if-block at line \ref{line:cancelinvesting} is never entered, then the investment stage is effectively identical to that of Algorithm \ref{alg:eDalg}. Therefore, we can use Lemma \ref{lem:genericusedbudgets}. The budget assigned to $v$ at $t_j$ must be $\theta_2 \cdot c(v)$ because $v$ is not disproportional at time $t_j$, which therefore implies that $c(I_{v, i}) \geq \theta_2 \cdot c(v)$.
\end{proof}

The following lemma generalizes Lemma \ref{lem:nodisproportionalbelow} to sets of vertices.

\begin{lemma} \label{lem:descendants}
    Fix a path $P \in \mathcal P$. Let $\mathcal U \subseteq A_i \cap P$ be such that there exists no recently disproportional vertices in $P$ below any vertex in $\mathcal U$. Then $c(\bigcup_{v \in \mathcal U} I_{v, i}) \geq \theta_2 \cdot c(\mathcal U)$.
\end{lemma}

\begin{proof}
    Observe that the lemma would be easy if all the $I_{v, i}$ were disjoint, because $c(I_{v, i}) \geq \theta_2 c(v)$ by Lemma \ref{lem:nodisproportionalbelow}. Therefore, we are most interested in the case where the $I_{v, i}$ overlap, motivating the following claim.

    \begin{claim} \label{claim:simultaneous}
        Suppose $v, v' \in \mathcal  U$ are such that $I_{v, i} \cap I_{v', i} \neq \varnothing$. As $v, v' \in A_i$, let $t_j$ and $t_{j'}$ be the most recent times before $t_i$ that $v$ and $v'$ were expansion vertices, respectively ($t_j$ and $t_{j'}$ are guaranteed to exist by Lemma \ref{lemma:expectedafterexpansion}). Then $t_j = t_{j'}$.
    \end{claim}

    \begin{proof}[Proof of claim.]
    Without loss of generality, assume that $v'$ is deeper than $v$, and because $v$ and $v'$ both lie on the same path in $\mathcal P$, this means that $v'$ is a descendant of $v$. Let $w \in I_{v, i} \cap I_{v', i}$ be arbitrary. $w$ must be a descendant of both $v$ and $v'$, by Observation \ref{obs:deeperinvestments}. Now, as $V^E$ always forms a subtree by Observation \ref{obs:expectedsubtree}, it is known that $t_{j} \geq t_{j'}$, as whenever $v'$ is anticipated, $v$ will also be anticipated. Thus it is sufficient to prove that $t_j > t_{j'}$ is impossible.

    Hence, seeking a contradiction, assume $t_j > t_{j'}$. This would imply that $v \in T_j^E$ but $v' \not \in T_j^E$. This requires that $v'$ must be a descendant of $v^\downarrow_{P(v), j}$. Thus $w$ is also a descendant of $v^\downarrow_{P(v), j}$, and thus it is only possible for $w \in I_{v, j}'=I_{v, i}$ if line \ref{line:cancelinvesting} triggers $I_{v, j}' = I_{v^\downarrow_{P(v), j}, j}'$. However, this would imply that during the expansion phase for $v$ at $t_i$, that $v^\downarrow_{P(v), j}$ would be added to $V^E_i$. However, as $v^\downarrow_{P(v), j}$ would then be at least as deep as $v \in \mathcal U$, this would contradict the definition of $\mathcal U$ in Lemma \ref{lem:descendants}, as $v^\downarrow_{P(v), j}$ would be recently disproportionate at time $t_i$. Thus $t_j = t_{j'}$, proving Claim \ref{claim:simultaneous}.
    \end{proof}

    Following this claim, we can partition $\mathcal U$ into disjoint sets $U_1, ..., U_m$ such that all vertices in a given $U_k$ were most recently expansion vertices before $t_i$ at the same time. Then, due to Claim \ref{claim:simultaneous}, if $v \in U_k$ and $v' \in U_{k'}$ with $k \neq k'$ then $I_{v, i} \cap I_{v', i} = \varnothing$. Finally, observe that $\theta_2 c(U_k) \leq c(\bigcup_{v \in U_k} I_{v, i})$, because $\theta_2 c(U_k)$ is the total budget assigned to vertices in $U_k$ when they were all expansion vertices at the latest time before $t_i$, and so $c(\bigcup_{v \in U_k} I_{v, i})$ can be no less than this budget because their investments occurred at the same timestamp in the algorithm. Putting this together, we obtain
    \begin{equation*}
        c(\bigcup_{v \in \mathcal U} I_{v, i}) = \sum_{k=1}^m c(\bigcup_{v \in U_k} I_{v, i}) \geq \sum_{k=1}^m \theta_2 c(U_k) =  \theta_2 \cdot c(\mathcal U).
    \end{equation*}
\end{proof}

The next lemma shows that for leaf nodes $P$ in $\mathcal H$, whenever there are anticipated vertices on $P$ there is also a recently disproportional vertex on $P$.

\begin{lemma} \label{lem:leafcase}
    Let $P$ be a leaf node in $\mathcal H$. Suppose $A_i \cap P$ is non-empty, and let $v^*$ be the deepest vertex in $A_i \cap P$. Then at time $t_i$, $v^*$ was recently disproportional. 
\end{lemma}
\begin{proof}
    We prove the contrapositive. Suppose instead that $v^* \in A_i \cap P$ was not recently disproportional. Let $t_j$ the the most recent time before $t_i$ that $v^*$ was an expansion vertex, using Lemma \ref{lemma:expectedafterexpansion}. As $v^*$ is not recently disproportional, we know $v_{P, j}^\downarrow$ is deeper than $v^*$. As $v^* \in A_i$, we know that $next'_{v^*, j} = next_{v^*, i} \neq \infty$, which implies that during the investment stage for $v^*$ at time $t_j$, that either its budget was expended, or if-block on line \ref{line:cancelinvesting} was entered. As $P$ is a leaf in $\mathcal{H}$, it must have been that the if-block on line \ref{line:cancelinvesting} was entered, because any $v'$ chosen by the algorithm during the investment phase for $v$ must be deeper than $v_{P, j}^\downarrow$. Thus we know $I_{v^*, i}=I_{v^*, j}' = I_{v_{P, j}^\downarrow, i}'$, and that these sets are non-empty as $next'_{v_{P, j}^\downarrow, j} = next'_{v^*, j} = next_{v^*, i} \neq \infty$. Therefore, $v_{P, j}^\downarrow$ must have been added to $V_i^E$ during the expansion stage for $v^*$ at time $t_i$. Additionally, as $next'_{v_{P, j}^\downarrow, j} = next'_{v^*, j} = next_{v^*, i} \leq t_i$ (where the inequality is due to $v^* \in A_i$), we know that $v_{P, j}^\downarrow \in A_i$ also. Therefore, $v^*$ is not the deepest vertex in $A_i \cap P$, as required.
\end{proof}

We now present the analogue of Lemma \ref{lem:coretounexpected}, modified for the heavy path decomposition's structure. The overall proof technique is broadly similar to that of \ref{lem:coretounexpected}. 

\begin{lemma} \label{lem:coretounexpected-h}
    For any $t_i$, we have
    \begin{equation*}
    c(V_i^E)  \leq \left(1+\frac{1}{\theta_1}\right)^{H+1} \left(1+\frac{1}{\theta_2}\right)^{H}c(\bar{A}_i)
\end{equation*}
\end{lemma}

\begin{proof}
    Fix $t_i$. We begin by placing coins with value $c(v)$ on all expansion vertices at $t_i$, thus the total cost equals the left-hand side of the lemma. We will show that redistributing these costs leads to an expression upper-bounded by the right-hand side.

    We say we \emph{process} a path $P$ of $\mathcal{H}$ when we redistribute the costs of the anticipated expansion vertices on $P$. We outline this processing as follows. We call a path $P$ in $\mathcal{H}$ \emph{ready to redistribute} if all of the paths $P'$ which are ancestors of $P$ in $\mathcal{H}$ have been processed. 
    
    When a path $P$ is ready to redistribute, we redistribute the costs on $P$ in 2 stages, after which we deem a path processed. Let $v^*$ be the deepest recently disproportional vertex in $P$ at time $t_i$, if any such vertices exist. If it exists, then in \textbf{Stage 1} of this reallocation process, we reallocate the costs on all anticipated nodes on $p(v^*)$ to $v^*$ itself. Then we reallocate the cost on $v^*$ to $I_{v^*, i}$ (and $I_{v^*, i} \subseteq V_i^E$, as $v^* \in A_i$). The effect of this is that if all nodes $v$ on $p(v^*)$ have a cost of $\leq C c(v)$, then the total amount reallocated is $\leq C \cdot c(p(v^*))$ and this cost is spread uniformly onto $I_{v^*, i}$ with $c(I_{v^*, i}) \geq \theta_1 \cdot c(p(v^*))$ by Lemma \ref{lem:disproportionalvertices}. Thus each node $w$ below $v^*$ receives an additional cost of at most $\frac{C}{\theta_1} c(w)$.

    If $P$ is a leaf node in $\mathcal H$, then by Lemma \ref{lem:leafcase}, after the Stage 1 there will be no more costs on anticipated vertices in $P$. Otherwise if $P$ is not a leaf node in $\mathcal H$, we proceed with \textbf{Stage 2}. Consider $\mathcal U := A_i \cap \{v \in P(v^*) : v \text{ is a descendant of } v^*\}$ if $v^*$ existed, otherwise define $\mathcal U := A_i$. In either case, $\mathcal{U}$ satisfies the requirements of \ref{lem:descendants}, because $v^*$ was defined as the deepest recently disproportional vertex at time $t_i$. We reallocate the costs on $\mathcal U$ to $\bigcup_{v \in \mathcal U} I_{v, i}$, which by Lemma \ref{lem:descendants} has $c(\bigcup_{v \in \mathcal U} I_{v, i}) \geq \theta_2 \cdot c(\mathcal U)$. The effect of this is that if each vertex $v \in \mathcal{U}$ had a cost of $\leq C' c(w)$ prior to the Stage 2 reallocation, then each vertex $w \in \bigcup_{v \in \mathcal U} I_{v, i}$ receives an additional cost of at most $\frac{C'}{\theta_2}$. Crucially, after the Stage 1 and Stage 2 reallocations are performed, there will be no more costs on anticipated vertices in $P$, as $P \cap \bigcup_{v \in \mathcal U} I_{v, i} = \varnothing$. If $P \cap \bigcup_{v \in \mathcal U} I_{v, i} \neq \varnothing$, then this would imply the existence of a recently disproportionate vertex at time $t_i$, as only a disproportional vertex $v$ could have invested onto $P(v)$ directly.

    We are now ready to prove the main claim of this lemma.

    \begin{claim} \label{claim:heavyinduction}

        When a node $P$ of $\mathcal{H}$ at depth $n$ is ready to redistribute, the costs of all nodes which belong to paths in $\mathcal{H}_P$ (i.e: $P$ or descendants of $P$ in $\mathcal H$) is bounded above by $(1 + 1/\theta_1)^n(1+1/\theta_2)^n c(v)$.
    
    \end{claim}
    
    \begin{proof}[Proof of claim]
     We proceed by induction. The case of $n=0$ is trivial, so we move to $n \geq 1$. Consider a path $P'$, at a depth $n+1$ in $\mathcal{H}$. If $P'$ is ready to redistribute, then the parent of $P'$ in $\mathcal H$, say $P$, must have been processed already, and when it was processed, the cost on all nodes of the paths in $\mathcal{H}_P$ must have been $\leq (1+ 1/\theta_1)^n(1+1/\theta_2)^n c(v)$ by the inductive hypothesis. After the Stage 1 processing of $P$ therefore, all nodes $w$ deeper than the deepest disproportional node $v^*$, receive an extra cost of at most $\frac{(1+ 1/\theta_1)^n(1+1/\theta_2)^n}{\theta_1}c(v)$ (by our earlier observation that if all nodes $v$ on paths in $\mathcal{H}_{P}$ currently has a cost $\leq C\cdot c(v)$, then the added cost to each node below $v^*$ is $\frac{C}{\theta_1} c(w)$). Immediately before the Stage 2 reallocation therefore, the costs on nodes deeper than $v^*$ is bounded by 
    \[\left(1+\frac{1}{\theta_1}\right)^{n}\left(1+\frac{1}{\theta_2}\right)^n c(v) + \frac{1}{\theta_1}\left(1+\frac{1}{\theta_1}\right)^{n}\left(1+\frac{1}{\theta_2}\right)^n c(v) = \left(1+\frac{1}{\theta_1}\right)^{n+1}\left(1+\frac{1}{\theta_2}\right)^nc(v).\] 
    
    Then, in Stage 2 redistribution for $P$, the costs of $\leq(1+1/{\theta_1})^{n+1}(1+ 1/\theta_2)^n c(v)$ on anticipated nodes of $P$ below $v^*$ are reallocated off of $P$, adding an extra cost of $\leq\frac{1}{\theta_2}(1+1/{\theta_1})^{n+1}(1+ 1/\theta_2)^{n} c(v)$ to nodes within paths below $P$ in $\mathcal{H}$,. Thus, the expansion nodes below $P$ in $\mathcal{H}$ have a total cost of at most 
    \[\left(1+\frac{1}{\theta_1}\right)^{n+1}\left(1+\frac{1}{\theta_2}\right)^n c(v) + \frac{1}{\theta_2}\left(1+\frac{1}{\theta_1}\right)^{n+1}\left(1+\frac{1}{\theta_2}\right)^n c(v) = \left(1+\frac{1}{\theta_1}\right)^{n+1}\left(1+\frac{1}{\theta_2}\right)^{n+1} c(v).\]
    \end{proof}
    
    Thereby proving Claim \ref{claim:heavyinduction}. Finally, on the paths corresponding to leaf nodes of $\mathcal H$, there will be no Stage 2 of investments, as was explained above. After Stage 1 of the processing of the leaf node paths, we are done, as all costs on anticipated vertices at time $t_i$ would have been moved onto vertices which are unanticipated at time $t_i$. Any unanticipated vertex $v$ ends up with a cost of at most $(1+\frac{1}{\theta_1})^{H+1}(1+\frac{1}{\theta_2})^Hc(v)$ after the final Stage 1 reallocation step on the leaf nodes. This corresponds to the right-hand side of the inequality of the lemma.
\end{proof}

Combining inequality \eqref{lem:alltoexpansion-h} and Lemma \ref{lem:coretounexpected-h} implies Lemma \ref{lem:amortized-h}.

\section{Conclusion}

This paper developed a memory-based framework for online MLAP-D, which allowed for a direct analysis in the case of general trees, and improved on competitiveness results. To show the framework was flexible, we demonstrated that it could be modified to provide competitive ratios linear in caterpillar dimension, a measure which unifies several results for special classes of graph which had been considered previously for MLAP-D. Future research could extend the memory-based framework to the problem of MLAP with delay costs, which would likely tighten the upper bounds of competitive ratios for this problem. The framework could also be modified to work for other special tree classes. The most notable open problem in this area remains the wide asymptotic gap for the competitive ratio, where the best-known lower bound remains constant.

\newpage

\bibliographystyle{acm}
\bibliography{reference}

\newpage

\appendix

\section{An Example for Algorithm \ref{alg:eDalg}} \label{sec:app-ex}

Consider the following tree and requests, where the tree and vertex costs are given in advance. The depth is 3, so we set $\theta = 3$. All requests arrive at time $0$ except for $\rho_9$, which arrives at time $6$.

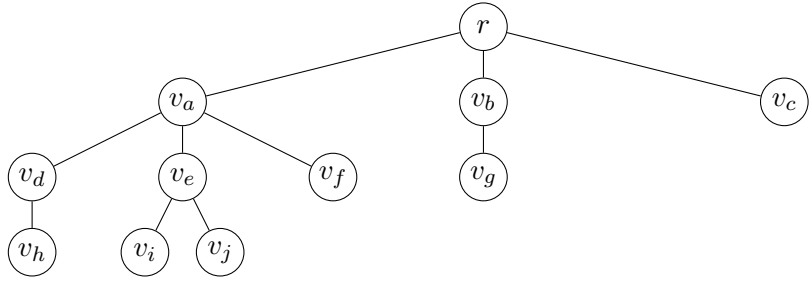
\begin{center}
    \begin{tikzpicture}
        \node[vertex] (a) at (0,5) {$r$};
        \node[vertex] (1) at (-4,4) {$v_a$};
        \node[vertex] (2) at (0,4) {$v_b$};
        \node[vertex] (3) at (4,4) {$v_c$};
        \node[vertex] (11) at (-6,3) {$v_d$};
        \node[vertex] (12) at (-4,3) {$v_e$};
        \node[vertex] (13) at (-2,3) {$v_f$};
        \node[vertex] (21) at (0,3) {$v_g$};
        \node[vertex] (111) at (-6,2) {$v_h$};
        \node[vertex] (121) at (-4.5,2) {$v_i$};
        \node[vertex] (122) at (-3.5,2) {$v_j$};
        \path (a) edge (1)
        (a) edge (2)
        (a) edge (3)
        (1) edge (11)
        (1) edge (12)
        (1) edge (13)
        (2) edge (21)
        (11) edge (111)
        (12) edge (121)
        (12) edge (122)
        ;
    \node [inner sep=0pt] (table) at (-2,0.5) {
         \begin{tabular}{c|C|C|C|C|C|C|C|C|C|C|C}
            vertex $v$ & $r$ & $v_a$ & $v_b$ & $v_c$ & $v_d$ & $v_e$ & $v_f$ & $v_g$ & $v_h$ & $v_i$ & $v_j$ \\ 
            \hline
            cost $c(v)$ & $1$ & $4$ & $1$ & $1$ & $1$ & $14$ & $7$ & $1$ & $1$ & $6$ & $60$
         \end{tabular}
   };
   \node [inner sep=0pt] (req) at (-2,-1.2) {
         \begin{tabular}{c|D|D|D|D|D|D|D|D|D}
            request $\rho$ & $\rho_1$ & $\rho_2$ & $\rho_3$ & $\rho_4$ & $\rho_5$ & $\rho_6$ & $\rho_7$ & $\rho_8$ & $\rho_9$ \\ 
            \hline
            vertex $v_\rho$ & $v_c$ & $v_b$ & $v_a$ & $v_e$ & $v_i$ & $v_j$ & $v_c$ & $v_g$ & $v_h$ \\
            \hline
            arrival $a_\rho$ & $0$ & $0$ & $0$ & $0$ & $0$ & $0$ & $0$ & $0$ & $6$ \\
            \hline
            deadline $d_\rho$ & $1$ & $2$ & $3$ & $4$ & $5$ & $8$ & $9$ & $10$ & $7$
         \end{tabular}
   };
    \end{tikzpicture}
\end{center}

Algorithm \ref{alg:eDalg} transmits the tree induced on $\{r,v_b,v_c\}$ and serves $\rho_1$, $\rho_2$, and $\rho_7$ at time $1$, transmits the tree induced on $\{r,v_a,v_b,v_g,v_e\}$ and serves $\rho_3$, $\rho_4$, and $\rho_8$ at time 3, transmits $r$-$v_a$-$v_e$-$v_i$ and serves $\rho_5$ at time 5, and transmits the tree
induced on $\{r,v_a,v_d,v_e,v_h,v_j\}$ and serves $\rho_6$ and $\rho_9$ at time 7. We describe the execution details.

\subsection*{Initialization ($t=0$)}

At time $t=0$, all information except $\rho_9$ is revealed. The following information is from the procedure \textsc{Initialization}.

\begin{center}
    \begin{tabular}{c|C|C|C|C|C|C|C|C|C|C|C}
            vertex $v$ & $r$ & $v_a$ & $v_b$ & $v_c$ & $v_d$ & $v_e$ & $v_f$ & $v_g$ & $v_h$ & $v_i$ & $v_j$ \\ 
            \hline
            $\ell_v$ & $1$ & $4$ & $1$ & $1$ & $1$ & $14$ & $7$ & $1$ & $1$ & $6$ & $60$ \\
            \hline
            $next_v$ & $\infty$ & $\infty$& $\infty$& $\infty$& $\infty$& $\infty$& $\infty$& $\infty$& $\infty$& $\infty$& $\infty$ \\
            \hline
            $I_v$ & $\varnothing$ & $\varnothing$ & $\varnothing$ & $\varnothing$ & $\varnothing$ & $\varnothing$ & $\varnothing$ & $\varnothing$ & $\varnothing$ & $\varnothing$ & $\varnothing$
         \end{tabular}
\end{center}

\subsection*{The First Transmission ($t=1$)}

At time $t=1$, request $\rho_1$ becomes critical, invoking \textsc{OnDeadline}. $V^E = \{r,v_c\}$. \textsc{OnDeadline} calls \textsc{ExpansionStage$(V^E,r)$}. In \textsc{ExpansionStage$(V^E,r)$}, since $t=1 \le \infty = I_r$, this directly leads to the function call \textsc{ExpansionStage$(V^E,v_c)$}. In \textsc{ExpansionStage$(V^E,v_c)$}, since $t=1 \le \infty = I_{v_c}$ and $v_c$ has no child, $V^E$ remains unchanged, and we return to \textsc{OnDeadline}.

$V^I$ is initialized to $\varnothing$ and \textsc{OnDeadline} calls \textsc{InvestmentStage$(V^I,V^E,r)$}, invoking the call \textsc{InvestmentStage$(V^I,V^E,v_c)$}. In \textsc{InvestmentStage$(V^I,V^E,v_c)$}, $v_c$ is assigned a budget $b_{v_c} = 3 c(v_c) = 3$. Since $c$ is a leaf in $\cT$, $V(\cT_{v_c}) \setminus (V^E \cup V^I) = \varnothing$, $next_c = \infty$, \textsc{InvestmentStage$(V^I,V^E,v_c)$} ends, and we return to \textsc{InvestmentStage$(V^I,V^E,r)$}.

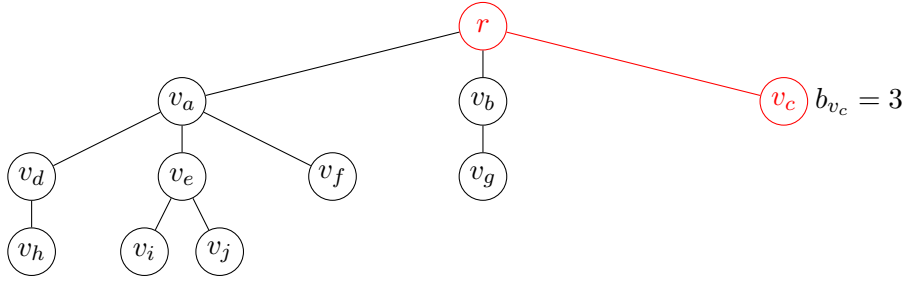
\begin{center}
    \begin{tikzpicture}
        \node[vertex,red] (a) at (0,5) {$r$};
        \node[vertex] (1) at (-4,4) {$v_a$};
        \node[vertex] (2) at (0,4) {$v_b$};
        \node[vertex,red] (3) at (4,4) {$v_c$};
        \node[vertex] (11) at (-6,3) {$v_d$};
        \node[vertex] (12) at (-4,3) {$v_e$};
        \node[vertex] (13) at (-2,3) {$v_f$};
        \node[vertex] (21) at (0,3) {$v_g$};
        \node[vertex] (111) at (-6,2) {$v_h$};
        \node[vertex] (121) at (-4.5,2) {$v_i$};
        \node[vertex] (122) at (-3.5,2) {$v_j$};
        \path (a) edge (1)
        (a) edge (2)
        (a) edge [red] (3)
        (1) edge (11)
        (1) edge (12)
        (1) edge (13)
        (2) edge (21)
        (11) edge (111)
        (12) edge (121)
        (12) edge (122)
        ;
        \node at (5,4) {$b_{v_c} = 3$};
    \end{tikzpicture}
\end{center}

In \textsc{InvestmentStage$(V^I,V^E,r)$}, $r$ is assigned a budget $b_r = 3 c(r) = 3$. We have that $V(\cT_r) \setminus (V^E \cup V^I) = V(\cT) \setminus \{r,v_c\}$. We invest toward the earliest pending request in $V(\cT) \setminus \{r,v_c\}$, that is, $\rho_2$ located at $v_b$. This invokes \textsc{Invest}$(r,v_b,V^I)$ where $V^I = \varnothing$. We invest 1 unit in $v_b$, so the budget $b_r=3-1 = 2$. $I_r$ is updated to $\{v_b\}$. Since the remaining cost $\ell_{v_b}$ becomes 0, we renew it so that $\ell_{v_b} = 1$. This accounts for serving potential future requests passing $v_b$. Finally, $V^I=\{v_b\}$ and $\rho_2$ is served. This ends the call of \textsc{Invest}$(r,v_b,V^I)$.

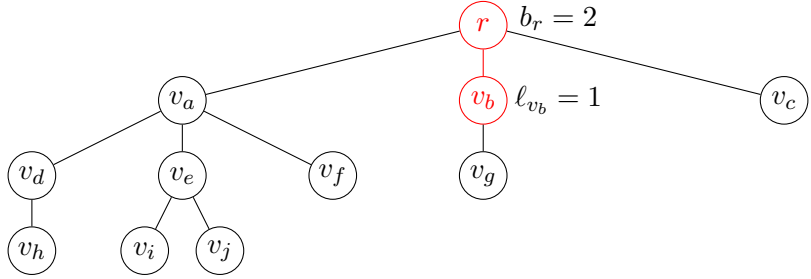
\begin{center}
    \begin{tikzpicture}
        \node[vertex,red] (a) at (0,5) {$r$};
        \node[vertex] (1) at (-4,4) {$v_a$};
        \node[vertex,red] (2) at (0,4) {$v_b$};
        \node[vertex] (3) at (4,4) {$v_c$};
        \node[vertex] (11) at (-6,3) {$v_d$};
        \node[vertex] (12) at (-4,3) {$v_e$};
        \node[vertex] (13) at (-2,3) {$v_f$};
        \node[vertex] (21) at (0,3) {$v_g$};
        \node[vertex] (111) at (-6,2) {$v_h$};
        \node[vertex] (121) at (-4.5,2) {$v_i$};
        \node[vertex] (122) at (-3.5,2) {$v_j$};
        \path (a) edge (1)
        (a) edge [red] (2)
        (a) edge (3)
        (1) edge (11)
        (1) edge (12)
        (1) edge (13)
        (2) edge (21)
        (11) edge (111)
        (12) edge (121)
        (12) edge (122)
        ;
        \node at (1,5.1) {$b_r = 2$};
        \node at (1,4) {$\ell_{v_b} = 1$};
    \end{tikzpicture}
\end{center}

Due to the update of $V^I$, $V(\cT_r) \setminus (V^E \cup V^I)$ becomes $V(\cT) \setminus \{r,v_b,v_c\}$. The earliest pending request in $V(\cT) \setminus \{r,v_b,v_c\}$ is $\rho_3$ located at $v_a$. This invokes \textsc{Invest}$(r,v_a,V^I)$. The entire budget $b_r=2$ is spent on $v_a$, making $\ell_{v_a} = 4-2 = 2$ and $I_r = \{v_a,v_b\}$. $\rho_3$ is still pending, so $next_r = d_{\rho_3} = 3$. We note that $V^I$ remains $\{v_b\}$ since $\rho_3$ is not served as $\ell_{v_a} > 0$.

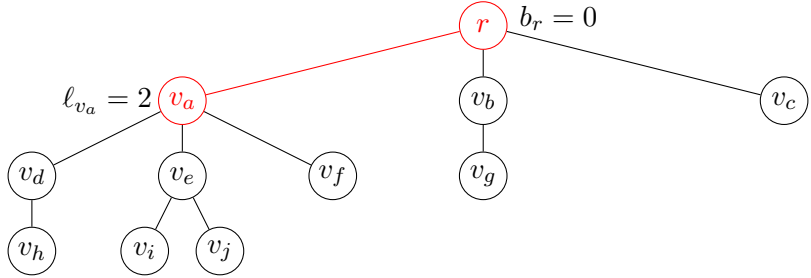
\begin{center}
    \begin{tikzpicture}
        \node[vertex,red] (a) at (0,5) {$r$};
        \node[vertex,red] (1) at (-4,4) {$v_a$};
        \node[vertex] (2) at (0,4) {$v_b$};
        \node[vertex] (3) at (4,4) {$v_c$};
        \node[vertex] (11) at (-6,3) {$v_d$};
        \node[vertex] (12) at (-4,3) {$v_e$};
        \node[vertex] (13) at (-2,3) {$v_f$};
        \node[vertex] (21) at (0,3) {$v_g$};
        \node[vertex] (111) at (-6,2) {$v_h$};
        \node[vertex] (121) at (-4.5,2) {$v_i$};
        \node[vertex] (122) at (-3.5,2) {$v_j$};
        \path (a) edge [red] (1)
        (a) edge (2)
        (a) edge (3)
        (1) edge (11)
        (1) edge (12)
        (1) edge (13)
        (2) edge (21)
        (11) edge (111)
        (12) edge (121)
        (12) edge (122)
        ;
        \node at (1,5.1) {$b_r = 0$};
        \node at (-5,4) {$\ell_{v_a} = 2$};
    \end{tikzpicture}
\end{center}

This ends the call of \textsc{InvestmentStage$(V^I,V^E,r)$} with $V^E = \{r,v_c\}$ and $V^I = \{v_b\}$. At $t=1$, the tree induced on $\{r,v_b,v_c\}$ is transmitted, $\rho_1$, $\rho_2$, and $\rho_7$ are served, and the information after \textsc{OnDeadline}s is updated as follows.

\begin{center}
    \begin{tabular}{c|c|C|C|C|C|C|C|C|C|C|C}
            vertex $v$ & $r$ & $v_a$ & $v_b$ & $v_c$ & $v_d$ & $v_e$ & $v_f$ & $v_g$ & $v_h$ & $v_i$ & $v_j$ \\ 
            \hline
            $\ell_v$ & $1$ & $2$ & $1$ & $1$ & $1$ & $14$ & $7$ & $1$ & $1$ & $6$ & $60$ \\
            \hline
            $next_v$ & $3$ & $\infty$& $\infty$& $\infty$& $\infty$& $\infty$& $\infty$& $\infty$& $\infty$& $\infty$& $\infty$ \\
            \hline
            $I_v$ & $\{v_a,v_b\}$ & $\varnothing$ & $\varnothing$ & $\varnothing$ & $\varnothing$ & $\varnothing$ & $\varnothing$ & $\varnothing$ & $\varnothing$ & $\varnothing$ & $\varnothing$
         \end{tabular}
\end{center}

\subsection*{The Second Transmission ($t=3$)}

At time $t=3$, request $\rho_3$ becomes critical which invokes \textsc{OnDeadline}, so $V^E = \{r,v_a\}$, $V^I = \varnothing$, and the call \textsc{ExpansionStage}$(V^E,r)$ is made. Note that $t = 3 = next_r$ and $v_b \in I_r$, we add $v_b$ to $V^E$ so $V^E = \{r,v_a,v_b\}$. The calls \textsc{ExpansionStage}$(V^E,v_a)$ and \textsc{ExpansionStage}$(V^E,v_b)$ are made. Since $v_a$ and $v_b$ are leaves in $V^E$, $V^E$ does not change after these functions are called. We return to \textsc{ExpansionStage}$(V^E,r)$ and it returns $V^E = \{r,v_a,v_b\}$ back to \textsc{OnDeadline}.

We have $V^I = \varnothing$ and \textsc{InvestmentStage}$(V^I,V^E,r)$ is called. The children of $r$ in $V^E$ are $v_a$ and $v_b$, so the calls \textsc{InvestmentStage}$(V^I,V^E,v_a)$ and \textsc{InvestmentStage}$(V^I,V^E,v_b)$ are made.

In \textsc{InvestmentStage}$(V^I,V^E,v_a)$, we have $I_{v_a} = \varnothing$, $b_{v_a} = 3c(v_a) = 12$, and $V(\cT_{v_a}) \setminus (V^E \cup V^I) = \{v_e,v_i,v_j\}$. The earliest pending request in $\{v_e,v_i,v_j\}$ is $\rho_4$ located at $v_e$. The call \textsc{Invest}$(v_a,v_e,V^I)$ is made. The budget $b_{v_a}$ is fully spent on $v_e$. We have $\ell_{v_e} = 14-12 = 2$ and $I_{v_a} = \{v_e\}$. This ends the call of \textsc{Invest}$(v_a,v_e,V^I)$. $\rho_4$ is still pending, so $next_{v_a} = d_{\rho_4} = 4$. \textsc{InvestmentStage}$(V^I,V^E,v_a)$ ends with $V^I = \varnothing$.

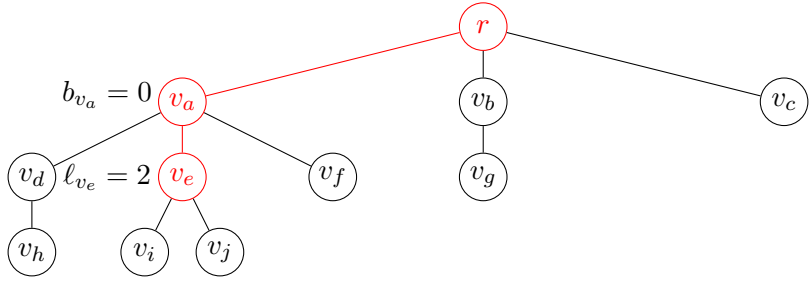
\begin{center}
    \begin{tikzpicture}
        \node[vertex,red] (a) at (0,5) {$r$};
        \node[vertex,red] (1) at (-4,4) {$v_a$};
        \node[vertex] (2) at (0,4) {$v_b$};
        \node[vertex] (3) at (4,4) {$v_c$};
        \node[vertex] (11) at (-6,3) {$v_d$};
        \node[vertex,red] (12) at (-4,3) {$v_e$};
        \node[vertex] (13) at (-2,3) {$v_f$};
        \node[vertex] (21) at (0,3) {$v_g$};
        \node[vertex] (111) at (-6,2) {$v_h$};
        \node[vertex] (121) at (-4.5,2) {$v_i$};
        \node[vertex] (122) at (-3.5,2) {$v_j$};
        \path (a) edge [red] (1)
        (a) edge (2)
        (a) edge (3)
        (1) edge (11)
        (1) edge [red] (12)
        (1) edge (13)
        (2) edge (21)
        (11) edge (111)
        (12) edge (121)
        (12) edge (122)
        ;
        \node at (-5,3) {$\ell_{v_e} = 2$};
        \node at (-5,4.1) {$b_{v_a} = 0$};
    \end{tikzpicture}
\end{center}

We proceed with \textsc{InvestmentStage}$(V^I,V^E,v_b)$. We have $I_{v_b} = \varnothing$, $b_{v_b} = 3 c(v_b) = 3$, and $V(\cT_{v_b}) \setminus (V^E \cup V^I) = \{v_g\}$. The earliest pending request in $\{v_g\}$ is $\rho_8$ located at $v_g$. The call \textsc{Invest}$(v_b,v_g,V^I)$ is made. We invest 1 unit in $v_g$, so the budget $b_{v_b} = 3-1 = 2$. $I_{v_b}$ is updated to $\{v_g\}$. Since the remaining cost $\ell_{v_g}$ becomes 0, we renew it so that $\ell_{v_g} = 1$. Finally, $V^I = \{v_g\}$ and $\rho_8$ is served. This ends the call of \textsc{Invest}$(v_b,v_g,V^I)$. There are no pending requests in $\{v_g\}$ so $next_{v_b} = \infty$. The call \textsc{InvestmentStage}$(V^I,V^E,v_b)$ ends with $V^I = \{v_g\}$.

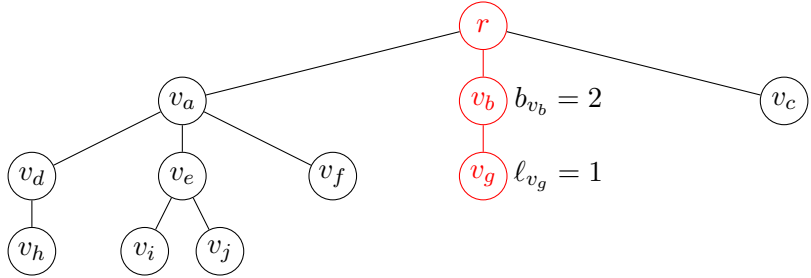
\begin{center}
    \begin{tikzpicture}
        \node[vertex,red] (a) at (0,5) {$r$};
        \node[vertex] (1) at (-4,4) {$v_a$};
        \node[vertex,red] (2) at (0,4) {$v_b$};
        \node[vertex] (3) at (4,4) {$v_c$};
        \node[vertex] (11) at (-6,3) {$v_d$};
        \node[vertex] (12) at (-4,3) {$v_e$};
        \node[vertex] (13) at (-2,3) {$v_f$};
        \node[vertex,red] (21) at (0,3) {$v_g$};
        \node[vertex] (111) at (-6,2) {$v_h$};
        \node[vertex] (121) at (-4.5,2) {$v_i$};
        \node[vertex] (122) at (-3.5,2) {$v_j$};
        \path (a) edge (1)
        (a) edge [red] (2)
        (a) edge (3)
        (1) edge (11)
        (1) edge (12)
        (1) edge (13)
        (2) edge [red] (21)
        (11) edge (111)
        (12) edge (121)
        (12) edge (122)
        ;
        \node at (1,4) {$b_{v_b} = 2$};
        \node at (1,3) {$\ell_{v_g} = 1$};
    \end{tikzpicture}
\end{center}

We return to \textsc{InvestmentStage}$(r)$ and have $I_r = \varnothing$, $b_r = 3 c(r) = 3$, and $V(\cT_r) \setminus (V^E \cup V^I) = V(\cT) \setminus \{r,v_a,v_b,v_g\}$. The earliest pending request in $V(\cT) \setminus \{r,v_a,v_b,v_g\}$ is $\rho_4$ located at $v_e$. The first vertex in $V(\cT) \setminus \{r,v_a,v_b,v_g\}$ on the path $r$-$v_e$ is $v_e$. The call \textsc{Invest}$(r,v_e,V^I)$ is made. Recall that $\ell_{v_e} = 2$, so we invest 2 unit in $v_e$ and $\rho_4$ is served. The budget $b_{r} = 3-2 = 1$. Since the remaining cost $\ell_{v_e}$ becomes 0, we renew it so that $\ell_{v_e} = 14$. $I_{r}$ is updated to $\{v_e\}$. $V^I = \{v_g,v_e\}$ since $\rho_8$ and $\rho_4$ have been served.

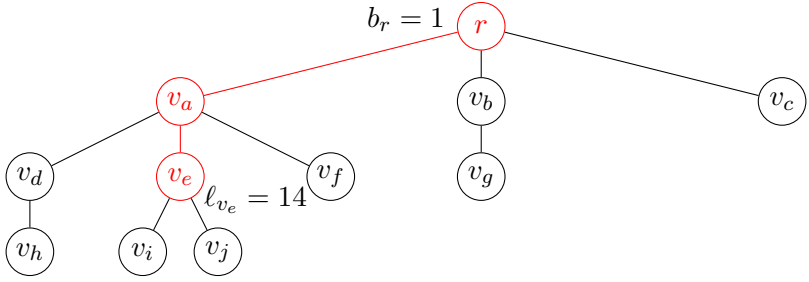
\begin{center}
    \begin{tikzpicture}
        \node[vertex,red] (a) at (0,5) {$r$};
        \node[vertex,red] (1) at (-4,4) {$v_a$};
        \node[vertex] (2) at (0,4) {$v_b$};
        \node[vertex] (3) at (4,4) {$v_c$};
        \node[vertex] (11) at (-6,3) {$v_d$};
        \node[vertex,red] (12) at (-4,3) {$v_e$};
        \node[vertex] (13) at (-2,3) {$v_f$};
        \node[vertex] (21) at (0,3) {$v_g$};
        \node[vertex] (111) at (-6,2) {$v_h$};
        \node[vertex] (121) at (-4.5,2) {$v_i$};
        \node[vertex] (122) at (-3.5,2) {$v_j$};
        \path (a) edge [red] (1)
        (a) edge (2)
        (a) edge (3)
        (1) edge (11)
        (1) edge [red] (12)
        (1) edge (13)
        (2) edge (21)
        (11) edge (111)
        (12) edge (121)
        (12) edge (122)
        ;
        \node at (-1,5.1) {$b_{r} = 1$};
        \node at (-3,2.7) {$\ell_{v_e} = 14$};
    \end{tikzpicture}
\end{center}

Due to the update of $V^I$, $V(\cT_r) \setminus (V^E \cup V^I)$ becomes $V(\cT) \setminus \{r,v_a,v_b,v_g,v_e\}$. The earliest pending request in $V(\cT) \setminus \{r,v_a,v_b,v_g,v_e\}$ is $\rho_5$ located at $v_i$. This invokes \textsc{Invest}$(r,v_i,V^I)$. The entire budget $b_r=1$ is spent on $v_i$, making $\ell_{v_i} = 6-1 = 5$ and $I_r = \{v_e,v_i\}$. $\rho_5$ is still pending, so $next_r = d_{\rho_5} = 5$. We note that $V^I$ remains $\{v_e,v_g\}$ since $\rho_5$ is not served as $\ell_{v_i} > 0$.

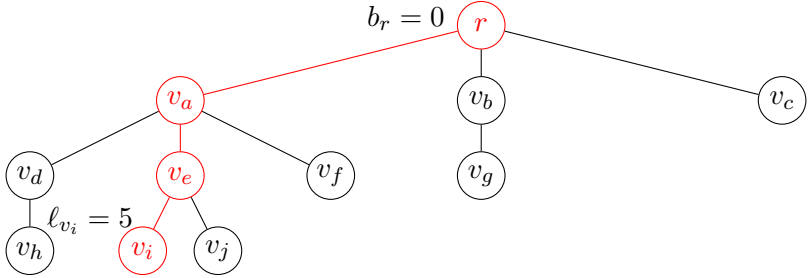
\begin{center}
    \begin{tikzpicture}
        \node[vertex,red] (a) at (0,5) {$r$};
        \node[vertex,red] (1) at (-4,4) {$v_a$};
        \node[vertex] (2) at (0,4) {$v_b$};
        \node[vertex] (3) at (4,4) {$v_c$};
        \node[vertex] (11) at (-6,3) {$v_d$};
        \node[vertex,red] (12) at (-4,3) {$v_e$};
        \node[vertex] (13) at (-2,3) {$v_f$};
        \node[vertex] (21) at (0,3) {$v_g$};
        \node[vertex] (111) at (-6,2) {$v_h$};
        \node[vertex,red] (121) at (-4.5,2) {$v_i$};
        \node[vertex] (122) at (-3.5,2) {$v_j$};
        \path (a) edge [red] (1)
        (a) edge (2)
        (a) edge (3)
        (1) edge (11)
        (1) edge [red] (12)
        (1) edge (13)
        (2) edge (21)
        (11) edge (111)
        (12) edge [red] (121)
        (12) edge (122)
        ;
        \node at (-1,5.1) {$b_{r} = 0$};
        \node at (-5.2,2.4) {$\ell_{v_i} = 5$};
    \end{tikzpicture}
\end{center}

This ends the call of \textsc{InvestmentStage$(V^I,V^E,r)$} with $V^E = \{r,v_a,v_b\}$ and $V^I = \{v_e,v_g\}$. At $t=3$, the tree induced on $\{r,v_a,v_b,v_g,v_e\}$ is transmitted, $\rho_3$, $\rho_4$, and $\rho_8$ are served, and the information after \textsc{OnDeadline} is updated as follows.

\begin{center}
    \begin{tabular}{c|c|C|C|C|C|C|C|C|C|C|C}
            vertex $v$ & $r$ & $v_a$ & $v_b$ & $v_c$ & $v_d$ & $v_e$ & $v_f$ & $v_g$ & $v_h$ & $v_i$ & $v_j$ \\ 
            \hline
            $\ell_v$ & $1$ & $2$ & $1$ & $1$ & $1$ & $14$ & $7$ & $1$ & $1$ & $5$ & $60$ \\
            \hline
            $next_v$ & $5$ & $4$ & $\infty$& $\infty$& $\infty$& $\infty$& $\infty$& $\infty$& $\infty$& $\infty$& $\infty$ \\
            \hline
            $I_v$ & $\{v_e,v_i\}$ & $\{v_e\}$ & $\{v_g\}$ & $\varnothing$ & $\varnothing$ & $\varnothing$ & $\varnothing$ & $\varnothing$ & $\varnothing$ & $\varnothing$ & $\varnothing$
         \end{tabular}
\end{center}

\subsection*{The Third Transmission ($t=5$)}

At time $t=5$, request $\rho_5$ becomes critical which invokes \textsc{OnDeadline}, so $V^E = \{r,v_a,v_e,v_i\}$, $V^I = \varnothing$, and the call \textsc{ExpansionStage}$(V^E,r)$ is made. Note that $t = 5 = next_r$ and $I_r = \{v_e, v_i\}$, we add $v_e$ and $v_i$ to $V^E$ and $V^E$ remains $\{r,v_a,v_e,v_i\}$. The call \textsc{ExpansionStage}$(V^E,v_a)$ is made. Since $t = 5 > 4 = next_{v_a}$ and $I_{v_a} = \{v_e\}$, we add $v_e$ to $V^E$ and $V^E$ remains $\{r,v_a,v_e,v_i\}$. The call \textsc{ExpansionStage}$(V^E,v_e)$ is made. Since $t = 5 < \infty = next_{v_e}$, $I_{v_e}$ remains $\varnothing$. The call \textsc{ExpansionStage}$(V^E,v_i)$ is made. Since $v_i$ is a leaf in $\cT$, \textsc{ExpansionStage}$(V^E,v_i)$ immediately ends, and so do \textsc{ExpansionStage}$(V^E,v_e)$, \textsc{ExpansionStage}$(V^E,v_a)$, and \textsc{ExpansionStage}$(V^E,r)$. 

We return to \textsc{OnDeadline}, have $V^I = \varnothing$, and \textsc{InvestmentStage}$(V^I,V^E,r)$ is called. The child node of $r$ in $V^E$ is $v_a$, so the call \textsc{InvestmentStage}$(V^I,V^E,v_a)$ is made.

In \textsc{InvestmentStage}$(V^I,V^E,v_a)$, we have $I_{v_a} = \{v_e\}$, so \textsc{InvestmentStage}$(V^I,V^E,v_e)$ is called. In \textsc{InvestmentStage}$(V^I,V^E,v_e)$, we have $I_{v_e} = \varnothing$, $b_{v_e} = 3c(v_e) = 42$ and $V(\cT_{v_e}) \setminus (V^E \cup V^I) = \{v_j\}$. The earliest pending request in $\{v_j\}$ is $\rho_6$ located at $v_j$. The call \textsc{Invest}$(v_e,v_j,V^I)$ is made. $b_{v_e}$ is fully spent on $v_j$, so $\ell_{v_j} = 60 - 42 = 18$ and $I_{v_e} = \{v_j\}$. $\rho_6$ is still pending, so $next_{v_e} = d_{\rho_6} = 8$. This ends the call of \textsc{Invest}$(v_e,v_j,V^I)$ with $V^I = \varnothing$ since $\rho_6$ is not served. This ends the call \textsc{InvestmentStage}$(V^I,V^E,v_e)$.

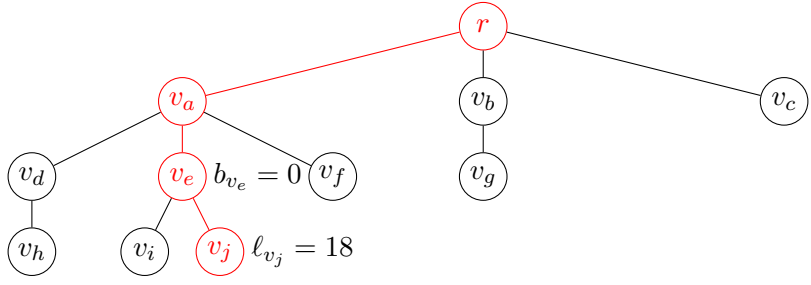
\begin{center}
    \begin{tikzpicture}
        \node[vertex,red] (a) at (0,5) {$r$};
        \node[vertex,red] (1) at (-4,4) {$v_a$};
        \node[vertex] (2) at (0,4) {$v_b$};
        \node[vertex] (3) at (4,4) {$v_c$};
        \node[vertex] (11) at (-6,3) {$v_d$};
        \node[vertex,red] (12) at (-4,3) {$v_e$};
        \node[vertex] (13) at (-2,3) {$v_f$};
        \node[vertex] (21) at (0,3) {$v_g$};
        \node[vertex] (111) at (-6,2) {$v_h$};
        \node[vertex] (121) at (-4.5,2) {$v_i$};
        \node[vertex,red] (122) at (-3.5,2) {$v_j$};
        \path (a) edge [red] (1)
        (a) edge (2)
        (a) edge (3)
        (1) edge (11)
        (1) edge [red] (12)
        (1) edge (13)
        (2) edge (21)
        (11) edge (111)
        (12) edge (121)
        (12) edge [red] (122)
        ;
        \node at (-3,3) {$b_{v_e} = 0$};
        \node at (-2.4,2) {$\ell_{v_j} = 18$};
    \end{tikzpicture}
\end{center}

We return to \textsc{InvestmentStage}$(V^I,V^E,v_a)$ and have $I_{v_a} = \varnothing$, $b_{v_a} = 3 c(v_a) = 12$, and $V(\cT_{v_a}) \setminus (V^E \cup V^I) = \{v_d,v_f,v_h,v_j\}$. The earliest pending request in $\{v_d,v_f,v_h,v_j\}$ is $\rho_6$ located at $v_j$, so \textsc{Invest}$(v_a,v_j,V^I)$ is called. $b_{v_a}$ is fully spent on $v_j$, so $\ell_{v_j} = 18 - 12 = 6$ and $I_{v_a} = \{v_j\}$. $\rho_6$ is still pending, so $next_a = d_{\rho_6} = 8$. This ends the call of \textsc{InvestmentStage}$(V^I,V^E,v_a)$ with $V^I = \varnothing$ since $\rho_6$ is not served.

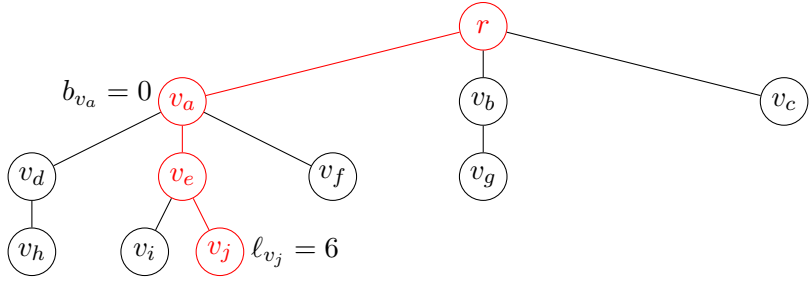
\begin{center}
    \begin{tikzpicture}
        \node[vertex,red] (a) at (0,5) {$r$};
        \node[vertex,red] (1) at (-4,4) {$v_a$};
        \node[vertex] (2) at (0,4) {$v_b$};
        \node[vertex] (3) at (4,4) {$v_c$};
        \node[vertex] (11) at (-6,3) {$v_d$};
        \node[vertex,red] (12) at (-4,3) {$v_e$};
        \node[vertex] (13) at (-2,3) {$v_f$};
        \node[vertex] (21) at (0,3) {$v_g$};
        \node[vertex] (111) at (-6,2) {$v_h$};
        \node[vertex] (121) at (-4.5,2) {$v_i$};
        \node[vertex,red] (122) at (-3.5,2) {$v_j$};
        \path (a) edge [red] (1)
        (a) edge (2)
        (a) edge (3)
        (1) edge (11)
        (1) edge [red] (12)
        (1) edge (13)
        (2) edge (21)
        (11) edge (111)
        (12) edge (121)
        (12) edge [red] (122)
        ;
        \node at (-2.5,2) {$\ell_{v_j} = 6$};
        \node at (-5,4.1) {$b_{v_a} = 0$};
    \end{tikzpicture}
\end{center}

We return to \textsc{InvestmentStage}$(V^I,V^E,r)$ and have $I_r = \varnothing$, $b_r = 3 c(r) = 3$, and $V(\cT) \setminus (V^E \cup V^I) = V(\cT) \setminus \{r,v_a,v_e,v_i\}$. The earliest pending request in $V(\cT) \setminus \{r,v_a,v_e,v_i\}$ is $\rho_6$ located at $v_j$, so \textsc{Invest}$(r,v_j,V^I)$ is called. $b_{r}$ is fully spent on $v_j$, so $\ell_{v_j} = 6 - 3 = 3$ and $I_{r} = \{v_j\}$. $\rho_6$ is still pending, so $next_r = d_{\rho_6} = 8$. This ends the call of \textsc{InvestmentStage}$(V^I,V^E,r)$ with $V^I = \varnothing$ since $\rho_6$ is not served.

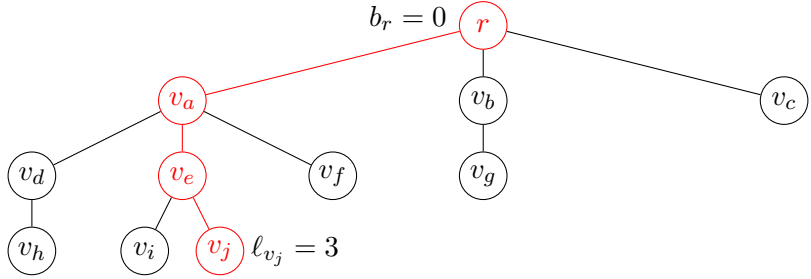
\begin{center}
    \begin{tikzpicture}
        \node[vertex,red] (a) at (0,5) {$r$};
        \node[vertex,red] (1) at (-4,4) {$v_a$};
        \node[vertex] (2) at (0,4) {$v_b$};
        \node[vertex] (3) at (4,4) {$v_c$};
        \node[vertex] (11) at (-6,3) {$v_d$};
        \node[vertex,red] (12) at (-4,3) {$v_e$};
        \node[vertex] (13) at (-2,3) {$v_f$};
        \node[vertex] (21) at (0,3) {$v_g$};
        \node[vertex] (111) at (-6,2) {$v_h$};
        \node[vertex] (121) at (-4.5,2) {$v_i$};
        \node[vertex,red] (122) at (-3.5,2) {$v_j$};
        \path (a) edge [red] (1)
        (a) edge (2)
        (a) edge (3)
        (1) edge (11)
        (1) edge [red] (12)
        (1) edge (13)
        (2) edge (21)
        (11) edge (111)
        (12) edge (121)
        (12) edge [red] (122)
        ;
        \node at (-2.5,2) {$\ell_{v_j} = 3$};
        \node at (-1,5.1) {$b_r = 0$};
    \end{tikzpicture}
\end{center}

The transmission tree returned is the path $r$-$v_a$-$v_e$-$v_i$, which serves requests $\rho_5$. The information after \textsc{OnDeadline} is updated as follows.

\begin{center}
    \begin{tabular}{c|C|C|C|C|C|C|C|C|C|C|C}
            vertex $v$ & $r$ & $v_a$ & $v_b$ & $v_c$ & $v_d$ & $v_e$ & $v_f$ & $v_g$ & $v_h$ & $v_i$ & $v_j$ \\ 
            \hline
            $\ell_v$ & $1$ & $2$ & $1$ & $1$ & $1$ & $14$ & $7$ & $1$ & $1$ & $5$ & $3$ \\
            \hline
            $next_v$ & $8$ & $8$& $\infty$& $\infty$& $\infty$& $8$& $\infty$& $\infty$& $\infty$& $\infty$& $\infty$ \\
            \hline
            $I_v$ & $\{v_j\}$ & $\{v_j\}$ & $\varnothing$ & $\varnothing$ & $\varnothing$ & $\{v_j\}$ & $\varnothing$ & $\varnothing$ & $\varnothing$ & $\varnothing$ & $\varnothing$
         \end{tabular}
\end{center}

\subsection*{The Fourth Transmission ($t=7$)}
We recall that at time $t=6$, $\rho_9$ arrives and its information is revealed.

At time $t=7$, request $\rho_9$ becomes critical which invokes \textsc{OnDeadline}, so $V^E = \{r,v_a,v_d,v_h\}$, $V^I = \varnothing$, and the call \textsc{ExpansionStage}$(V^E,r)$ is made. Note that $t = 7 < 8 = next_r = next_{v_a}$. Without considering $I_r$, this invokes \textsc{ExpansionStage}$(V^E,v_a)$, \textsc{ExpansionStage}$(V^E,v_d)$, and \textsc{ExpansionStage}$(V^E,v_h)$. $v_h$ is a leaf node in $\cT$, so \textsc{ExpansionStage}$(V^E,v_h)$ ends immediately. \textsc{ExpansionStage}$(V^E,v_d)$, \textsc{ExpansionStage}$(V^E,v_a)$, and \textsc{ExpansionStage}$(V^E,r)$ also end immediately.

We return to \textsc{OnDeadline}, have $V^I = \varnothing$, and \textsc{InvestmentStage}$(V^I,V^E,r)$ is called. The child node of $r$ in $V^E$ is $v_a$, so the call \textsc{InvestmentStage}$(V^I,V^E,v_a)$ is made. Similarly, \textsc{InvestmentStage}$(V^I,V^E,v_a)$ calls \textsc{InvestmentStage}$(V^I,V^E,v_d)$ and \textsc{InvestmentStage}$(V^I,V^E,v_d)$ calls \textsc{InvestmentStage}$(V^I,V^E,v_h)$. \textsc{InvestmentStage}$(V^I,V^E,v_h)$ ends immediately and so does \textsc{InvestmentStage}$(V^I,V^E,v_d)$ as there are no pending requests below them.

In \textsc{InvestmentStage}$(V^I,V^E,v_a)$, we have $I_{v_a} = \varnothing$, $b_{v_a} = 3c(v_a) = 3$ and $V(\cT_{v_a}) \setminus (V^E \cup V^I) = \{v_e,v_f,v_i,v_j\}$. The earliest pending request in $\{v_e,v_f,v_i,v_j\}$ is $\rho_6$ located at $v_j$. The call \textsc{Invest}$(v_a,v_j,V^I)$ is made. $b_{v_a}$ is fully spent on $v_j$, so $\ell_{v_j} = 3 - 3 = 0$ and $\rho_6$ is served. Since the remaining cost $\ell_{v_j}$ becomes $0$, we renew it so that $\ell_{v_j} = 60$. $I_{v_a}$ is updated to $\{v_j\}$. This ends the call of \textsc{Invest}$(v_a,v_j,V^I)$ with $V^I = \{v_j\}$.

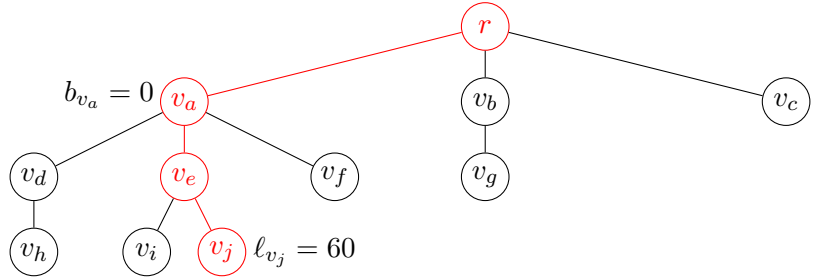
\begin{center}
    \begin{tikzpicture}
        \node[vertex,red] (a) at (0,5) {$r$};
        \node[vertex,red] (1) at (-4,4) {$v_a$};
        \node[vertex] (2) at (0,4) {$v_b$};
        \node[vertex] (3) at (4,4) {$v_c$};
        \node[vertex] (11) at (-6,3) {$v_d$};
        \node[vertex,red] (12) at (-4,3) {$v_e$};
        \node[vertex] (13) at (-2,3) {$v_f$};
        \node[vertex] (21) at (0,3) {$v_g$};
        \node[vertex] (111) at (-6,2) {$v_h$};
        \node[vertex] (121) at (-4.5,2) {$v_i$};
        \node[vertex,red] (122) at (-3.5,2) {$v_j$};
        \path (a) edge [red] (1)
        (a) edge (2)
        (a) edge (3)
        (1) edge (11)
        (1) edge [red] (12)
        (1) edge (13)
        (2) edge (21)
        (11) edge (111)
        (12) edge (121)
        (12) edge [red] (122)
        ;
        \node at (-2.4,2) {$\ell_{v_j} = 60$};
        \node at (-5,4.1) {$b_{v_a} = 0$};
    \end{tikzpicture}
\end{center}

There are no pending requests in $V(\cT_{v_a}) \setminus (V^E \cup V^I) = \{v_e,v_f,v_i\}$, so $next_{v_a} = \infty$. This ends the call \textsc{InvestmentStage}$(V^I,V^E,v_a)$ and we return to \textsc{InvestmentStage}$(V^I,V^E,r)$. There are no pending requests in $V(\cT)$, so $next_{r} = \infty$. This ends the call \textsc{InvestmentStage}$(V^I,V^E,r)$.The algorithm transmits the tree
induced on $\{r,v_a,v_d,v_e,v_h,v_j\}$ and serves $\rho_6$ and $\rho_9$. The information after \textsc{OnDeadline} is updated as follows.

\begin{center}
    \begin{tabular}{c|C|C|C|C|C|C|C|C|C|C|C}
            vertex $v$ & $r$ & $v_a$ & $v_b$ & $v_c$ & $v_d$ & $v_e$ & $v_f$ & $v_g$ & $v_h$ & $v_i$ & $v_j$ \\ 
            \hline
            $\ell_v$ & $1$ & $2$ & $1$ & $1$ & $1$ & $14$ & $7$ & $1$ & $1$ & $5$ & $60$ \\
            \hline
            $next_v$ & $\infty$ & $\infty$& $\infty$& $\infty$& $\infty$& $8$& $\infty$& $\infty$& $\infty$& $\infty$& $\infty$ \\
            \hline
            $I_v$ & $\varnothing$ & $\varnothing$ & $\varnothing$ & $\varnothing$ & $\varnothing$ & $\{v_j\}$ & $\varnothing$ & $\varnothing$ & $\varnothing$ & $\varnothing$ & $\varnothing$
         \end{tabular}
\end{center}

\section{Proof of Theorem \ref{thm:eD}} \label{app:pf-eD}

\thmeD*

\begin{proof}
Recall from Lemmas \ref{lem:amortized} and \ref{lem:optlowerbound}, we have that
\[\textsc{alg} \leq \left(1+\frac{1}{\theta}\right)^{D} (1+\theta)\textsc{opt}.\]

We now prove that the function $R(\theta) = (1+1/\theta)^D(1+\theta)$ is minimized at $\theta^* = D$. Because $\ln$ is a strictly increasing function it suffices to minimize \[\ln R(\theta) = D \ln\left(1+\frac{1}{\theta}\right) + \ln(1+\theta) = D (\ln(\theta+1) - \ln \theta) + \ln(1+\theta) = (D+1)\ln(1+\theta)-D\ln \theta.\] To find stationary points we solve \[ \frac{d}{d\theta} \ln R(\theta) = \frac{1+D}{1+\theta} - \frac{D}{\theta} = 0\]
    and so the only stationary point of $R$ will be at $\theta = D$. Observing also that $R(\theta) \rightarrow \infty$ as $\theta \rightarrow 0^+$, and also if $\theta \rightarrow \infty$ demonstrates that $\theta^*=D$ is a local minimum of $R$.

    Thus our analysis has proven we have an $R(D) = (1+\frac{1}{D})^D(D+1) \leq e(D+1)$ competitive algorithm for MLAP-D.
\end{proof}

\section{Proof of Theorem \ref{thm:4eH}} \label{app:pf-4eH}

\thmeH*

\begin{proof}
    Recall from Lemmas \ref{lem:amortized-h} and \ref{lem:optlowerbound}, for any $\theta_1, \theta_2 > 0$,
\begin{align*}
    \textsc{alg} 
    &\leq (1+\theta_1+\theta_2)\left(1+\frac{1}{\theta_1}\right)^{H+1}\left(1+\frac{1}{\theta_2}\right)^H \textsc{opt}.
\end{align*}

    We wish to show that $\theta_1=2H+1$ and $\theta_2=2H$ is the best solution for minimizing the competitive ratio. Let $R(\theta_1, \theta_2) = (1+\theta_1+\theta_2)(1+1/{\theta_1})^{H+1}(1+1/\theta_2)^H$. As $\ln$ is strictly increasing, we can choose to minimize $\ln R(\theta_1, \theta_2)$ instead. We have
    \begin{align*}
        \ln R(\theta_1, \theta_2) &= \ln(1+\theta_1+\theta_2) + (H+1)\ln(1+1/{\theta_1}) +H \ln(1+1/\theta_2),\\
        \frac{\partial \ln R(\theta_1, \theta_2)}{\partial {\theta_1}} &= \frac{1}{1+\theta_1+\theta_2} - \frac{H+1}{\theta_1 + \theta_1^2} = 0,\\
        \frac{\partial \ln R(\theta_1, \theta_2)}{\partial {\theta_2}} &= \frac{1}{1+\theta_1+\theta_2} - \frac{H}{\theta_2 + \theta_2^2} = 0.
    \end{align*}

    Solving for $H$ in both equations yields
    \begin{align*}
        H &= \frac{\theta_1 + \theta_1^2}{1+\theta_1+\theta_2}-1 = \frac{\theta_2 + \theta_2^2}{1+\theta_1+\theta_2}\\
        &\Rightarrow \theta_1 + \theta_1^2 - 1 - \theta_1 - \theta_2 = \theta_2 + \theta_2^2\\
        &\Rightarrow \theta_1^2= \theta_2^2+2\theta_2+1=(\theta_2 + 1)^2
    \end{align*}
    and $\theta_1 = \theta_2 + 1$ as $\theta_1, \theta_2 > 0$. Thus,
    \begin{equation*}
        H = \frac{\theta_1+\theta_1^2}{2\theta_1}-1 = \frac{\theta_1}{2}-\frac{1}{2}
        \Rightarrow \theta_1
        = 2H+1
    \end{equation*}
    and $\theta_2 = 2H$. Therefore, Algorithm \ref{alg:heavyalg} has a competitive ratio bounded above by $(4H+2)(1+\frac{1}{2H+1})^{H+1}(1+\frac{1}{2H})^H$. Finally, we prove that 
    \begin{align}
        \left(1+\frac{1}{2H+1}\right)^{H+1}\left(1+\frac{1}{2H}\right)^H \le e \label{ineq:e}
    \end{align}
for all integers $H > 0$, thereby proving that the competitive ratio is bounded by $(4H+2)e$.

Let $f(H):=R(2H+1,2H)$. We show that \eqref{ineq:e} holds when $H=1$ and $H \ge 2$.

When $H=1$, we have that
\[\left(\frac{4}{3}\right)^{2} \cdot \frac{3}{2} = \frac{8}{3} \le e.\]

When $H \ge 2$, we show that $f(H) \to e$ as $H \to \infty$ and $f(H)$ is increasing in $H$. When $H \to \infty$, we have that
\begin{align*}
f(H) &=  \left(1+\frac{1}{2H+1}\right)^{H+1/2} \left(\frac{2H+2}{2H+1}\right)^{1/2} \left(1+\frac{1}{2H}\right)^H \\
&= e^{1/2} \left(\frac{2H+2}{2H+1}\right)^{1/2} e^{1/2} \to e \text{ as } H \to \infty.
\end{align*}
To show that $f(H)$ is increasing in $H$, we show that $\ln f(H)$ is increasing in $H$.
\begin{align*}
    \ln f(H) &= (H+1) \ln \frac{2H+2}{2H+1} + H \ln \frac{2H+1}{2H} \\
    &= (H+1) \ln (2H+2) + (H - H - 1) \ln (2H+1) + H \ln (2H) \\
    &= (H+1) \ln (2H+2) - \ln (2H+1) - H \ln (2H).
\end{align*}
Taking the derivative, we wish to show that
\begin{align*}
    \frac{d}{d H} \ln f(H) &= \ln (2H+2) + \frac{H+1}{2H+2} \cdot 2 - \frac{2}{2H+1} - \ln (2H) - \frac{H}{2H} \cdot 2 \\
    &= \ln (2H+2) - \ln (2H) - \frac{2}{2H+1} \\
    &= \ln \left(1+\frac{1}{H}\right) - \frac{2}{2H+1} > 0
\end{align*}
when $H \ge 2$. By taking the Maclaurin series of $\ln(1+1/H)$ up to the fourth order, we have that
\begin{align*}
    \ln \left(1+\frac{1}{H}\right) &\ge \frac{1}{H} - \frac{1}{2H^2} + \frac{1}{3H^3} - \frac{1}{4H^4}\\
\end{align*}
so it is sufficient to show that
\[\frac{1}{H} - \frac{1}{2H^2} + \frac{1}{3H^3} - \frac{1}{4H^4} = \frac{12H^3 - 6H^2 + 4H -3}{12H^4} \ge \frac{2}{2H+1}.\]
This holds because
\begin{align*}
    &\frac{12H^3 - 6H^2 + 4H -3}{12H^4} \ge \frac{2}{2H+1} \\
    \iff & (2H+1) \left(12H^3 - 6H^2 + 4H -3\right) \ge 24H^4 \\
    \iff & 24H^4 + 2H^2 - 2H -3 \ge 24 H^4 \\
    \iff & 2H^2 - 2H -3 \ge 0
\end{align*}
whenever $H \ge 2$.
\end{proof}

\end{document}